\documentclass[11pt, oneside]{amsart}   	
\usepackage{geometry}                		
\usepackage{graphicx}				
\usepackage{subcaption}								
\usepackage{amssymb}
\usepackage{bm}
\usepackage{hyperref}
\usepackage{amsmath}
\usepackage{tikz-cd,tikz,tkz-tab}
\usepackage{multirow}
\usepackage[utf8]{inputenc}
\usepackage{amsthm}
\usepackage{caption}

\newcommand{\func}[5]{{#1} : \begin{array}{ccc} {#2} &\longrightarrow& {#3} \\ {#4} &\longmapsto& {#5} \end{array}}
\newcommand{\dd}{\textrm{d}}
\newcommand{\kstar}{{}^*}

\newtheorem{definition}{Definition}
\newtheorem{rem}{Remark}
\newtheorem{lemme}{Lemma}
\newtheorem{prop}{Proposition}
\newtheorem{corollaire}{Corollary}

\usepackage{mathrsfs}

\usepackage{placeins}

\let\Oldsection\section
\renewcommand{\section}{\FloatBarrier\Oldsection}

\let\Oldsubsection\subsection
\renewcommand{\subsection}{\FloatBarrier\Oldsubsection}

\let\Oldsubsubsection\subsubsection
\renewcommand{\subsubsection}{\FloatBarrier\Oldsubsubsection}

\hypersetup{colorlinks,citecolor=blue,filecolor=black,linkcolor=red, urlcolor=magenta}

\title{Maximal Kerr-de Sitter spacetimes}
\author{Jack Borthwick}
\thanks{Author affiliations: LMBA, UMR CNRS 6205, Department of Mathematics, University of Brest, 6 avenue Victor Le Gorgeu,
29200 Brest, France. Email: jack.borthwick@univ-brest.fr}
\date{\today}							

\begin{document}
\maketitle
\begin{abstract}
In this note, we propose a survey of the basic geometric properties of Carter's Kerr-de Sitter solution to Einstein's equation with cosmological constant. In particular, we give simple characterisations of the Kerr-de Sitter analogs of fast, slow and extreme Kerr spacetime and conclude with a discussion on maximal analytical extensions in each of these cases.
\end{abstract}

\tableofcontents 
\section{Introduction}
Over the past decade or so, there has been increasing interest in asymptotically de Sitter spacetimes, as opposed to the well-studied asymptotically flat spacetimes, notably due to the experimental evidence that our universe is actually in expansion, and that this expansion is accelerating. De Sitter spacetime, named after the Dutch mathematician and astronomer Willem de Sitter, is one of the simpler models of such a universe. It can be seen as the submanifold of equation $\displaystyle -x_0^2 + \sum_i^{n+1} x_i^2 =\alpha^2 ,\alpha \in \mathbb{R}$ in $(n+2)$-dimensional Minkowski space and is a maximally symmetric vacuum solution to Einstein's equation with positive cosmological constant $\Lambda=\frac{3}{\alpha^2}$; the parameter $\alpha$ is also related to the Ricci scalar by $R=\frac{n(n-1)}{\alpha^2}$. In this paper, we are interested in 4-dimensional Kerr-de Sitter spacetimes describing a rotating black hole on a de-Sitter background. These solutions where first discussed by Brandon Carter~\cite{Carter:2009aa}, but more thorough studies of them, and in particular of the structure of the roots of the polynomial $\Delta_r$ according to the values of the parameters $a,l$ and $M$, have been delayed, until recently, due to its supposed more geometrical than physical significance. In recent articles, several authors have shown interest in Kerr-de Sitter spacetimes, and a numerical study is proposed in~\cite{Matzner:2011aa}. 

In this work we give complete and relatively simple characterisations of the Kerr-de Sitter analogs of ``fast'', ``extreme'' and ``slow'' Kerr spacetime and describe in detail the construction of a maximal analytical extension of the Kerr-de Sitter solution in each case. The text is organised as follows: in section~\ref{section:kds_metric} we give a succinct description of the geometric properties of the Kerr-de Sitter metric in Carter's Boyer-Lindquist like coordinates; the principal result of interest is the computation of the curvature forms $\Omega^{i}_{\,\,j}$. Following~\cite{Gibbons:1977aa,Matzner:2011aa}, the sign convention for $\Lambda$ is opposite to that in Carter's original work. In section~\ref{delta_r}, we discuss the root structure of the family of polynomials $\Delta_r$ according to the values of the parameters $(a,l,M)$. After writing this article, we discovered that a similar study had already been lead in~\cite{Lake:2015aa}; our results confirm and complete theirs. In section~\ref{maximal}, we describe the construction of maximal Kerr-de Sitter spacetimes, the criterion for maximality being the completeness of all principal null geodesics that do not run into a curvature singularity. The results of section~\ref{section:kds_metric} confirm the fact that only minor adaptations of the methods used in~\cite{ONeill:2014aa} are required, however, some of the proofs are repeated and complements are provided in appendices so that the text is as self-contained as possible. 
We decided not to discuss more general geodesics than the principal nulls used in the construction of maximal extensions, but found that recent articles had ventured into this terrain: a classification of null geodesics is proposed in~\cite{Charbulak:2017aa} and a discussion on all causal geodesics is given in ~\cite{Zannias:2017aa}.

The signature convention used in this work is $(-,+,+,+)$ and, when units are relevant, formulae are written in geometric units where $G=1$ and $c=1$.
\section{The Kerr-de Sitter metric}
\label{section:kds_metric}
In this section we will define the Kerr-de Sitter $(KdS)$ metric $g$ and calculate the curvature forms $\Omega^i_{\,\,j}$ on each of the so-called ``Boyer-Lindquist blocks'' in an appropriate frame. The algebraic structure of the curvature tensor encoded in these forms will show that, like that of the Kerr metric, the Weyl tensor of the Kerr-de Sitter metric is of Petrov type D at each point of these blocks. 

The components $g_{ij}$ of the Kerr-de Sitter metric on the connected components of the manifold $(\mathbb{R}_t\times \mathbb{R}_r)\times S^2 \setminus \Sigma \cup \mathcal{H}$, $\mathcal{H}=\{\Delta_r=0\}, \Sigma=\{\rho^2=0\}$, referred to as the Boyer-Lindquist (BL) blocks, are given in table~\ref{BLmetric}; some useful alternative expressions are also given in appendix~\ref{app:diverse}. When $l=0$, these expressions reduce to those of the usual Kerr metric. The coordinates $(t,r,\theta,\phi)$ will be referred to as Boyer-Lindquist(-like) coordinates.

\begin{table}[h]

\centering

\begin{tabular}{|c|c|c|}
\hline
& Kerr metric & Kerr-de Sitter Metric \\
\hline
$g_{tt}$ & $-1 +\frac{2rM}{\rho^2}$ & $\frac{\Delta_\theta a^2\sin^2\theta - \Delta_r}{\rho^2 \Xi^2}$ \\ 
\hline
$g_{rr}$ & $\frac{\rho^2}{\Delta}$ & $\frac{\rho^2}{\Delta_r}$ \\
\hline
$g_{\theta\theta}$ & $\rho^2$ & $\frac{\rho^2}{\Delta_\theta}$ \\
\hline
$g_{\phi\phi}$ &$ \left[r^2 +a^2  +\frac{2rMa^2\sin^2\theta}{\rho^2} \right]\sin^2\theta $ & $\left[ \Delta_\theta(r^2+a^2)^2 - \Delta_ra^2\sin^2\theta  \right]\frac{\sin^2\theta}{\rho^2 \Xi^2}$\\ \hline
$g_{\phi t} $ &$ -\frac{2rMa\sin^2\theta}{\rho^2}$ &$ \frac{a\sin^2\theta}{\Xi^2 \rho^2}\left( \Delta_r - \Delta_\theta(r^2+a^2) \right) $ \\\hline
Other & All zero & All zero
\\\hline
\multicolumn{3}{|c|}{$l^2=\frac{\Lambda}{3}\quad  \quad \Xi=1+l^2a^2 \quad \quad  \Delta_\theta=1+l^2a^2\cos^2\theta $}\\\multicolumn{3}{|c|}{$\Delta_r =\Delta - l^2r^2(r^2+a^2) \quad \quad  \rho^2=r^2+a^2\cos^2\theta \quad \quad \Delta=r^2-2Mr+a^2 $}\\\hline 
\end{tabular}
\caption{\label{BLmetric}Metric tensor elements in Boyer-Lindquist like coordinates}
\end{table}

The parameters $a,M$ and $\Lambda$ have their usual physical interpretation: $M$ is the mass of the black hole, $a$ its angular momentum per unit mass and $\Lambda$ is the cosmological constant,

As in the case of the Kerr metric, the Kerr-de Sitter metric line element can be divided into two parts that clearly have an unique analytic extension to all of $(\mathbb{R}_t\times \mathbb{R}_r)\times S^2 \setminus \Sigma \cup \mathcal{H}$ (whereas the expressions in table~\ref{BLmetric} are a priori only valid at points where $\sin \theta \neq 0$). 

More precisely we have $ \dd s^2 = g_{rr}\dd r^2 + Q + Q' $ where $Q$ and $Q'$ are the two quadratic forms given by:

\begin{align}Q &= g_{tt}\text{d}t^2  + 2 g_{\phi t} \dd \phi \dd t\\\nonumber &= -\frac{\Delta_\theta}{\Xi^2}\dd t^2 + \frac{1}{\Xi^2}\left( l^2(r^2+a^2) + \frac{2Mr}{\rho^2} \right)\left(\left[\dd t - a\sin^2\theta \dd \phi \right]^2-a^2\sin^4\theta \dd \phi^2\right) \end{align}

\begin{align} 
Q' &= g_{\theta\theta}\dd \theta^2 + g_{\phi\phi}\dd\phi^2 = \frac{\rho^2}{\Delta_\theta} \dd \sigma ^2 +\left(\frac{\Xi}{\Delta_\theta}\left(1-l^2r^2\right) +\frac{2Mr}{\rho^2}\right)\frac{a^2\sin^4\theta}{\Xi^2}\dd \phi^2
\end{align}


In the last expression $\dd \sigma^2 = \dd \theta^2 +\sin^2\theta \dd \phi^2$ is the usual line element of the sphere, which is naturally extendable to the poles. Moreover, the form $a\sin^2\theta \dd \phi$ is well defined\footnote{In cartesian coordinates it is $a(x\dd y -y\dd x)$} on all of $S^2$. Hence, the above expressions have unique analytic extensions to the points of the ``axis" $\mathcal{A}=\mathbb{R}^2 \times \{p_{\pm}\}$ where $p_\pm$ are the poles of the sphere.

The set $\Sigma$ is the ring singularity of the Kerr-de Sitter spacetime and the zeros of $\Delta_r$ will give us the number of Boyer-Lindquist blocks as well as the position of the horizons when we construct a maximal analytical extension of the Boyer-Lindquist blocks in section~\ref{maximal}. Its sign will also be of importance since, as seen from the expression in table~\ref{BLmetric}, it determines the nature\footnote{ space-like $g(v,v)>0$, time-like $g(v,v)<0$, light-like or isotropic $g(v,v)=0$} of the coordinate vector fields $\partial_t, \partial_r, \partial_\phi$. The properties of $\Delta_r$ will be studied in section~\ref{delta_r}. For now, we write $\varepsilon = \textrm{sgn}(\Delta_r)$ and define an orthonormal frame $(E_i)_{i\in \{0,\dots,3\} }$ on each Boyer-Lindquist block as follows:
\begin{equation}
\label{can_frame}
\begin{split}
 E_0= \frac{V\Xi}{\rho\sqrt{\varepsilon\Delta_r}} \quad  E_1 = \frac{\sqrt{\varepsilon\Delta_r}}{\rho} \partial_r \\ E_2=\frac{\sqrt{\Delta_\theta}}{\rho}\partial_\theta \quad  E_3= \frac{\Xi W}{\sin\theta \sqrt{\Delta_\theta}\rho}
 \end{split}
\end{equation}
The choice of vector fields $V=(r^2+a^2)\partial_t + a \partial_\phi $ and $W=\partial_\phi + a \sin^2 \theta \partial_t $ to replace $\partial_t$ and $\partial_\phi$ reduces the indeterminacy of the nature of the vectors to the sign of $\Delta_r$ which will be constant on each Boyer-Lindquist block. It is identical to that in~\cite{ONeill:2014aa} for the Kerr metric, where they play an important role; this will also be the case for the Kerr-de Sitter metric.

The dual frame is readily determined from~\eqref{can_frame}:

\begin{align}
 \omega^0&= \frac{\sqrt{\varepsilon \Delta_r }}{\Xi\rho} \textrm{d}t - \frac{a \sin^2\theta \sqrt{\varepsilon\Delta_r}}{\rho\Xi} \textrm{d}\phi & \omega^1 &= \frac{\rho}{\sqrt{\varepsilon\Delta_r}}\textrm{d}r \\\omega^3&= \frac{(r^2+a^2)\sqrt{\Delta_\theta}\sin\theta}{\rho\Xi}\textrm{d}\phi - \frac{a \sqrt{\Delta_\theta} \sin \theta}{\rho \Xi} \textrm{d}t & \omega^2&= \frac{\rho}{\sqrt{\Delta_\theta}} \textrm{d}\theta 
\end{align}

This furnishes a more compact expression of the line element:

\begin{align} \label{eq:metricBL_compact} \dd s^2&= -\varepsilon (\omega^0)^2 + \varepsilon (\omega^1)^2 + (\omega^2)^2 + (\omega^3)^2 \\\nonumber
=& - \frac{\Delta_r}{\Xi^2\rho^2}\left[ \dd t - a\sin^2\theta\textrm{d}\phi \right]^2 +\frac{\rho^2}{\Delta_r} \dd r^2 + \frac{\rho^2}{\Delta_\theta}\dd \theta^2 +\frac{\Delta_\theta\sin^2\theta}{\rho^2\Xi^2}\left[ (r^2+a^2)\dd \phi - a \dd t \right]^2 
\end{align}

From these expressions one can determine the connexion forms\footnote{given in appendix~\ref{app_connection_forms}} $v \mapsto \omega^i_{\,\, j}(v)=\omega^i(\nabla_v E_j)$, characterised uniquely by the first structural equation $ \dd \omega^i = -\sum_{m} \omega^i_{\,\,m}\wedge \omega^m $, and the curvature forms $\Omega^i_{\,\,j}= d\omega^i_{\,\, j} + \sum_{m} \omega^i_{\,\,m}\wedge \omega^m_{\,\, j}$.  The curvature forms are:

\begin{equation}
\label{eq:curvature_forms}
\begin{split}
\Omega^0_{\,\,\, 1} &= \varepsilon(2I + l^2) \omega^0 \wedge \omega^1 + 2 \varepsilon J \omega^3\wedge\omega^2 \\
\Omega^0_{\,\,\, 2} &=  -\varepsilon J \omega^1 \wedge \omega ^3 + (I-l^2) \omega^2\wedge \omega ^0 \\
\Omega^0_{\,\,\, 3} &=  \varepsilon J \omega^1\wedge \omega^2 - (I-l^2)\omega^0 \wedge \omega^3\\ 
\Omega^1_{\,\,\, 2} &=  -(I-l^2)\omega^1\wedge \omega^2 - \varepsilon J\omega^0\wedge \omega^3 \\
\Omega^1_{\,\,\, 3} &= -(I-l^2)\omega^1 \wedge \omega^3 + \varepsilon J \omega^0\wedge \omega^2\\
\Omega^2_{\,\,\, 3} &= 2J\omega^0\wedge \omega^1 +(2I+l^2)\omega^2\wedge \omega^3
\end{split}
\end{equation}

where: $I=\frac{Mr}{\rho^6}(r^2-3a^2\cos^2\theta)$ and $J=\frac{Ma\cos\theta}{\rho^6}(3r^2-a^2\cos^2\theta)$. When $l=0$ these formulae coincide with those in~\cite{ONeill:2014aa}\footnote{It should be noted that there is a small error in the expression of $\Omega^0_3$ given on page 98 of~\cite{ONeill:2014aa}, it should read: $\Omega^0_{\,\,\, 3} =  - I\omega^0 \wedge \omega^3 \bm{+}\varepsilon J \omega^1\wedge \omega^2$}. It is surprising to find that the additional contribution due to the presence of a positive cosmological constant $\Lambda$ is completely separate from that of the curvature due to the black hole.

The curvature forms are related to the Riemann curvature tensor by:

\begin{equation} \omega^a(R(E_c,E_d)E_b) =R^a_{\,\,\,bcd}=\Omega^a_{\,\,\,b}(E_c,E_d) \end{equation}
As in the case of Kerr metric, the presence of the factor $\rho^{-6}$ in these formulae indicates that the loci of $\rho^2=0$ is a real curvature singularity and that there is no sensible extension of the Boyer-Lindquist block containing $\Sigma$ to include these points. Using~\eqref{eq:curvature_forms} we find that the Ricci tensor is given by: 
\begin{equation} R_{ab}= 3l^2g_{ab}=\Lambda g_{ab} \end{equation} and so the Kerr-de Sitter metric is indeed a vacuum solution to Einstein's field equations with cosmological constant: \begin{equation} R_{ab} - \frac{1}{2}Rg_{ab} +\Lambda g_{ab}=0 \end{equation}

The relative simplicity of~\eqref{eq:curvature_forms} is reflected in the algebraic decomposition of the Riemann curvature tensor. In particular, we find that the Weyl conformal tensor\footnote{ $C_{abcd} = R_{abcd} - \frac{1}{2}\left( g_{ac}R_{bd} - g_{ad}R_{bc} + R_{ac}g_{bd} - R_{ad}g_{bc}\right) + \frac{R}{6}(g_{ac}g_{bd}-g_{ad}g_{bc})$} is given by:

\begin{equation}C_{abcd} = R_{abcd} - l^2(g_{ac}g_{bd} -g_{ad}g_{bc})\end{equation}

We can deduce from this that the conformal properties of the KdS-Boyer-Lindquist blocks are exactly those of the Kerr Boyer-Lindquist blocks $(l=0)$. In particular:
\begin{prop}
\label{prop:kds_conformal_properties}
\begin{enumerate}
\item  At each point of the Boyer-Lindquist blocks the Weyl tensor has Petrov type D
\item The principal null directions are determined by the rays of $E_0 \pm E_1$ or equivalently, $\pm \partial_r + \frac{\Xi}{\Delta_r}V$
 \end{enumerate}
\end{prop}

\begin{rem}
The normalisation chosen here is different from that in~\cite{Matzner:2011aa}, our choice is justified by the following lemma. 
\end{rem}

\begin{lemme}
On each Boyer-Lindquist block the integral curves of $\pm \partial_r +\frac{\Xi}{\Delta_r}V$ are geodesics.
\end{lemme}

\begin{proof}
This is actually a consequence of the Petrov type of C\footnote{cf. Goldberg-Sachs theorem~\cite{Goldberg:2009aa}}, but since we have at our disposition all of the connection forms, we can also verify it directly. The geodesic equations are given in appendix~\ref{app:geodesic_equations}. Consider an integral curve $\gamma : I \mapsto KdS$ of $\partial_r + \frac{\Xi}{\Delta_r}V$. It satisfies for $t\in I$: 
\begin{equation} \dot{\gamma}(t) = \frac{\rho}{\sqrt{\varepsilon \Delta_r}}|_{\gamma(t)} E_1(t) + \frac{\varepsilon\rho}{\sqrt{\varepsilon \Delta_r}}|_{\gamma(t)} E_0(t) \end{equation} Setting $\Gamma^3=\Gamma^2=0$ in the left-hand side of the equations in the appendix, shows that the last one is trivial and the remaining three reduce to:
\begin{align} \label{geodesic1}\dot{\Gamma^0}(t)&=- \frac{\partial}{\partial r} \left.\left(\frac{\sqrt{\varepsilon\Delta_r}}{\rho}\right)\right|_{\gamma(t)} \Gamma^0(t)\Gamma^1(t)  \\ \dot{\Gamma^1}(t)&=-\frac{\partial}{\partial r} \left.\left(\frac{\sqrt{\varepsilon\Delta_r}}{\rho}\right) \right|_{\gamma(t)} \left( \Gamma^0(t)\right)^2 \label{geodesic2} \\ \label{geodesic3}(\Gamma^0(t))^2&=(\Gamma^1(t))^2\end{align}

Equation~\eqref{geodesic3} is clearly satisfied and, substituting the expressions of $\Gamma^0$ and $\Gamma^1$ into the right-hand side of~\eqref{geodesic1} equation, we find:
\begin{align*} - \left.\frac{\partial}{\partial r} \left(\frac{\sqrt{\varepsilon\Delta_r}}{\rho}\right)\right|_{\gamma(t)} \Gamma^0(t)\Gamma^1(t) = -\varepsilon \left.\frac{\partial}{\partial r} \left(\frac{\sqrt{\varepsilon\Delta_r}}{\rho}\right)\frac{\rho^2}{\varepsilon \Delta_r}\right|_{\gamma(t)} &= \left.\varepsilon \frac{\partial}{\partial r} \left( \frac{\rho}{\sqrt{\varepsilon\Delta_r}}\right)\right|_{\gamma(t)}\\&= \dd r_{\gamma(t)}(\dot{\gamma}(t))\left.\varepsilon \frac{\partial}{\partial r} \left( \frac{\rho}{\sqrt{\varepsilon\Delta_r}}\right)\right|_{\gamma(t)} \\&= \dot{\Gamma^0}(t)\end{align*}

Similarly, for the right-hand side of~\eqref{geodesic2}:

\begin{equation*}\begin{split}-\frac{\partial}{\partial r} \left.\left(\frac{\sqrt{\varepsilon\Delta_r}}{\rho}\right) \right|_{\gamma(t)} \left( \Gamma^0(t)\right)^2=-\frac{\partial}{\partial r} \left.\left(\frac{\sqrt{\varepsilon\Delta_r}}{\rho}\right) \frac{\rho^2}{\varepsilon \Delta_r} \right|_{\gamma(t)}&=\dd r_{\gamma(t)}(\dot{\gamma}(t))\frac{\partial}{\partial r} \left.\left(\frac{\rho}{\sqrt{\varepsilon\Delta_r}}\right) \right|_{\gamma(t)}\\&=\dot{\Gamma^1}(t)\end{split}\end{equation*}
The remaining case is similar.
\end{proof}

\section{Fast, Extreme and Slow Kerr-de Sitter}
\label{delta_r}
In this section we study the structure of the roots of the family of polynomials:
 \begin{equation} \Delta_r(a,l,M)=r^2-2Mr+a^2-l^2r^2(r^2+a^2) \end{equation}
Throughout the following discussion we will assume that all of the parameters are non-zero, this guarantees that we are really on a de Sitter background and excludes Schwarzchild-de Sitter which is studied in \cite{Matzner:2011aa}. Moreover, we assume $a>0$, $l>0$. There is no loss of generality in assuming $a>0$ as all of the results of this section remain valid under the substitution $a\leftrightarrow |a|$, alternatively, we can always reverse the orientation of the axis of rotation. The restriction $l\neq 0$ also guarantees that $\text{deg}{\Delta_r}=4$. In the analytical extensions constructed in section~\ref{maximal}, each root of $\Delta_r$ will give rise to a totally geodesic null hypersurface, that we will refer to as a horizon.

Under the hypothesis that $l\neq 0$, it is clear that :

\begin{equation} \Delta_r = r^2-2Mr+a^2 -  l^2r^2(r^2+a^2) = 0 \Leftrightarrow r^4-\frac{1-l^2a^2}{l^2}r^2 + 2\frac{M}{l^2}r - \frac{a^2}{l^2}=0 \end{equation}

To simplify notations we introduce $ A= \frac{a}{l} $ and $m^2=\frac{M}{l^2}$, and will therefore study the structure of the roots of the degree 4 polynomial with real coefficients:

\begin{equation} P = X^4 -\frac{1-l^4A^2}{l^2}X^2 + 2m^2X-A^2 \label{def_poly} \end{equation}

Let us call $(x_1,x_2,x_3,x_4)$ the (not necessarily distinct) complex roots of $P$.
Writing out the Vieta formulae for this polynomial we know that the roots of $P$ must satisfy the following system:

\begin{equation}
\left\{\begin{array}{lc}
x_1+x_2+x_3+x_4 =0& (i) \\
x_1x_2+x_1x_3+x_1x_4+x_2x_3+x_2x_4+x_3x_4=\frac{A^2l^4-1}{l^2}&(ii)\\ 
x_1x_2x_3 +x_1x_2x_4+x_1x_3x_4+x_2x_3x_4=-2m^2&(iii)\\
x_1x_2x_3x_4=-A^2&(iv)
\end{array}\right. \label{Viete}\end{equation}

We can deduce immediately from equation $(iv)$ that for all positive real values of the parameters $A,m^2,l$  the polynomial $P$ will always have at least two distinct \emph{real} roots with opposite sign; these are the cosmological horizons. In particular, there is always a horizon ``inside" the singularity ($r<0$). Moreover, the multiplicity of any root is at most 3 and there is at most one root with multiplicity $>1$ 

\subsection{Extreme Kerr-de Sitter}

For the usual Kerr metric, extreme Kerr corresponds to the case where the polynomial $\Delta_r$ has a double root, i.e. the two black hole horizons coincide. A necessary and sufficient condition for this is that $M^2=a^2$. In this section we characterise the analogous case for the KdS metric. In fact, we find that there are three cases where horizons coincide: 

\begin{enumerate}
\item Three horizons situated in the region $r>0$ coincide.
\item The two black hole horizons coincide. 
\item The outer black hole horizon coincides with the outer cosmological horizon.
\end{enumerate}

We begin by proving the following proposition:

\begin{prop}
Let $a,M,l \in \mathbb{R}^*_+$ and P be defined by \eqref{def_poly}.
P has a root with multiplicity exactly 2 if and only if the parameters satisfy both of the following conditions:

\begin{itemize}
\item[$(i)$] $ al < 2-\sqrt{3}$
\item[$(ii)$] $\displaystyle M^2=\frac{(1-a^2l^2)(a^4l^4+34a^2l^2+1) \pm \sqrt{\delta} }{54l^2}$ 
\end{itemize}
\begin{gather*} \delta= (al-(2-\sqrt{3}))^3(al+2+\sqrt{3})^3(al+2-\sqrt{3})^3(al-(2+\sqrt{3}))^3\end{gather*}

Furthermore: [P has a root with multiplicity 3] $\Leftrightarrow \left\{\begin{array}{c} al=2-\sqrt{3}  \\ M^2=\frac{16}{9}\sqrt{3}a^3l\end{array}\right. $
\end{prop}
\begin{proof}

Firstly, a necessary and sufficient condition for the polynomial $P$ to have a root with multiplicity $>1$ is that its discriminant, $\Delta(P)$, should vanish. We recall that the discriminant is related to the resultant\footnote{The definition of the resultant is recalled in appendix~\ref{app_resultant}} $R(P,P')$ of $P$ and its formal derivative $P'$ by:  \begin{equation} \Delta(P)=\frac{(-1)^{\frac{n(n-1)}{2}}}{a_n}R(P,P') \end{equation}

In the above formula, $n$ is the degree of the polynomial, and $a_n$ is the coefficient of the leading term. Here: 
\begin{align*} \Delta(P)=\begin{split}&-\frac{16}{l^{10}}\left(a^{10} l^{8} + 4 \, a^{8} l^{6} + 6 \, a^{6} l^{4} + 27 \, M^{4} l^{2} + 4 \, a^{4} l^{2} \right.\\&\hspace{2in}+\left. {\left(a^{6} l^{6} + 33 \, a^{4} l^{4} - 33 \, a^{2} l^{2} - 1\right)} M^{2} + a^{2}\right) \end{split} \\ = &- \frac{16}{l^{10}}\left( 27M^4l^2 +(a^2l^2-1)(a^4l^4+34a^2l^2+1)M^2 +a^2(a^2l^2+1)^4 \right) \end{align*}

Thus:

\begin{equation} \Delta(P) = 0 \Leftrightarrow 27M^4l^2 +(a^2l^2-1)(a^4l^4+34a^2l^2+1)M^2 +a^2(a^2l^2+1)^4 =0  \end{equation}

This is a second order polynomial equation in $M^2$. We require that the roots be real and at least one of the roots be positive. However, as $a^2(a^2l^2+1)^4>0$ if one root is positive both of them are. Moreover, since the sum of the roots is given by $-(a^2l^2-1)(a^4l^4+34a^2l^2+1)$ when the roots exist and are real, they are both positive if and only if $al<1$.

The solutions are real if and only if the discriminant $\delta$ of the order two polynomial $Q=27X^2l^2 +(a^2l^2-1)(a^4l^4+34a^2l^2+1)X +a^2(a^2l^2+1)^4$ is positive. We find that:  
\begin{gather*} \delta= \left[ (1-a^2l^2)(a^4l^4+34a^2l^2 +1) - 6\sqrt{3} al(a^2l^2+1)^2 \right] \hspace{2.5in} \\ \hspace{2.5in} \times\left[(1-a^2l^2)(a^4l^4+34a^2l^2 +1) + 6\sqrt{3} al(a^2l^2+1)^2\right] \end{gather*}
Assuming as necessary $al< 1$ we see that $\delta$ has the same sign as:

$$\phi(al)=(1-a^2l^2)(a^4l^4+34a^2l^2 +1) - 6\sqrt{3} al(a^2l^2+1)^2 $$

Defining $y=al$, we are therefore interested in the sign of $\phi(y)$ for $y\in ]0,1[$. One can check\footnote{either by direct calculation or assuming simply $a^2l^2+2\sqrt{3}al-1=0$ } that $2-\sqrt{3}$ and $2+\sqrt{3}$ are a roots of $\phi$ and that $$\phi(y)=-(y-(2-\sqrt{3}))^3(y+2+\sqrt{3})^3$$
For $y\geq 0$, we find that $\phi(y)$ has opposite sign to $y-(2-\sqrt{3})$ and so is positive if and only if $y\leq (2-\sqrt{3})<1$.

Therefore, we have shown that $P$ has a root with multiplicity $>1$ if and only if \hbox{$al \leq (2-\sqrt{3})$} and $\displaystyle M^2=\frac{(1-a^2l^2)(a^4l^4+34a^2l^2+1) \pm \sqrt{\delta} }{54l^2} $.

We will now show that when $P$ has a root with multiplicity $>1$ it is of multiplicity $3$ if and only if $al= (2-\sqrt{3})$.

Suppose now that $P$ has a root $x$ with multiplicity $>1$. In particular the above conditions are satisfied. $x$ is of multiplicity at least two, and so, we can assume $x_3=x_4=x$. Vieta's formulae \eqref{Viete} then reduce to:
\begin{equation} \left\{\begin{array}{lc}x_1+x_2=-2x & (i') \\
x_1x_2-3x^2=\frac{A^2l^4-1}{l^2}&(ii')\\ 
x_1x_2x-x^3=-m^2&(iii')\\
x_1x_2x^2=-A^2&(iv') \end{array}\right. \label{V21} \end{equation}

Equation $(iv')$ show that as $A>0$ no root is zero so the system \eqref{V21} is equivalent to:
\begin{equation} \left\{\begin{array}{lc}x_1+x_2=-2x & (i') \\
3x^4 +\frac{A^2l^4-1}{l^2}x^2 +A^2=0&(ii'')\\ 
x^4-m^2x +A^2=0&(iii'')\\
x_1x_2x^2=-A^2&(iv') \end{array}\right. \label{V2} \end{equation}

Finally combining $(ii'')$ and $(iii'')$ we see that $\eqref{V2}$ is equivalent to:

\begin{equation} \left\{\begin{array}{lc}x_1+x_2=-2x & (i') \\
\frac{A^2l^4-1}{l^2}x^2+3m^2x -2A^2=0&(ii''')\\ 
x^4-m^2x +A^2=0&(iii'')\\
x_1x_2x^2=-A^2&(iv') \end{array}\right. \label{V3} \end{equation}

We assume now that $al=2-\sqrt{3}$. It follows that $\delta=0$, furthermore, noting that $a^2l^2+2\sqrt{3}al -1 =0$, it is straightforward to verify that:

\begin{equation} a^4l^4 +34a^2l^2 +1 = 48a^2l^2 \end{equation}

And therefore: \begin{equation} M^2 = \frac{16}{9}a^3l \sqrt{3} \label{triple_root_M}\end{equation}

Consider now $(ii''')$, which, written in terms of $a$ is:

\begin{equation} \frac{a^2l^2-1}{l^2}x^2 +3m^2x -2\frac{a^2}{l^2}=0  \end{equation}

We find that the equation has one double root given by:
\begin{equation} x=\frac{m^2l\sqrt{3}}{4a}\end{equation}

Now, the other two roots $x_1,x_2$, are the roots of the polynomial $$R=X^2 -(x_1+x_2)X+x_1x_2$$ By $\eqref{V3}$ one has: 

\begin{equation} R= X^2 +2xX-\frac{a^2}{l^2x^2} \end{equation}
The reduced discriminant $\delta'$ of $R$ is given by:
$$ \delta' = x^2 +\frac{a^2}{l^2x^2} $$

Since $$ x^2=\frac{3}{16}m^4\frac{l^2}{a^2}=\frac{\sqrt{3}}{3}\frac{a^3}{l^3}\frac{l^2}{a^2}=\frac{\sqrt{3}}{3}\frac{a}{l}$$

it follows that: $$\frac{1}{x^2}\frac{a^2}{l^2} = \sqrt{3} \frac{l}{a} \frac{a^2}{l^2} = \sqrt{3} \frac{a}{l} = 3 x^2$$

Hence: $\delta'=4x^2$ and the roots of $R$ are $x$ and $-3x$. The roots of $P$ and their multiplicities are then $(x,3),(-3x,1)$.

Conversely, assume that $P$ has a root of multiplicity $3$, say, without loss of generality: $x_1=x$ and $x_2=x_3=x_4=y$, Vieta's formulae \eqref{Viete} reduce this time to:

\begin{equation} \left\{\begin{array}{lc}x=-3y & (a) \\
\frac{A^2l^4-1}{l^2} = 3xy +3y^2&(b)\\ 
3xy^2+y^3=-2m^2&(c)\\
xy^3=-A^2&(d) \end{array}\right. \label{V4} \end{equation}

As before, equation $(d)$ forbids that one of the roots be zero so \eqref{V4} is equivalent to:

\begin{equation}
\left\{\begin{array}{lc}x=-3y & (a') \\
6y^2=\frac{1-A^2l^4}{l^2}  &(b')\\ 
4y^3=m^2&(c')\\
3y^4=A^2&(d') \end{array}\right. \label{V5}
\end{equation}

Equation $(c')$ shows that $y^3>0$ and so $y>0$ too, hence equation $(b')$ gives: 

$$ y = \frac{\sqrt{1-A^2l^4}}{\sqrt{6}l}$$

Equations $(c')$ and $(d')$ are compatibility equations, using the expression for $y$ we find that:

\begin{align}
A^2&=\frac{1}{12 l^4} (1-A^2l^4)^2  \label{A_constraint} \\
m^2&=\frac{2}{3}\frac{(1-A^2l^4)}{\sqrt{6}l^3}\sqrt{1-A^2l^4} \label{M_constraint}
\end{align}

As $m^2>0$ there is no loss of information in squaring \eqref{M_constraint} to find that:
$$m^4 =\frac{2}{27}\frac{(1-A^2l^4)^3}{l^2}$$

Or, in terms of $M$ and $a$:

\begin{equation} M^2 = \frac{2}{27}\frac{(1-A^2l^4)^3}{l^2} \label{triple_root_M_2} \end{equation}

Expanding \eqref{A_constraint} yields a second order equation for $A^2$:
\begin{equation}12A^2l^4=(1-A^2l^4)^2 \Leftrightarrow (A^2l^4 +2\sqrt{3}Al^2 -1)(A^2l^4-2\sqrt{3}Al^2-1)=0 \end{equation}

The equation $0=A^2l^4-2\sqrt{3}Al^2-1=a^2l^2-2\sqrt{3}al-1$ cannot give any solutions compatible with the condition $al\leq 2-\sqrt{3}<1$ as in this case
$$a^2l^2= 2\sqrt{3}al +1 \geq 1 $$ Consequently, we consider only the solutions of $A^2l^4+2\sqrt{3}Al^2-1 =0$. They are $A\in \{\frac{2-\sqrt{3}}{l^2}, -\frac{2+\sqrt{3}}{l^2}\}$. As we assume $A>0$ the second solution is excluded so $A$ must equal $\frac{2-\sqrt{3}}{l^2}$ which gives:

\begin{equation} al = 2-\sqrt{3} \end{equation}

Using the equation $a^2l^2+2\sqrt{3}al -1= 0$ we see that \eqref{triple_root_M_2} becomes:

\begin{equation} M^2=\frac{2}{27}\frac{(1-A^2l^4)^3}{l^2}= \frac{2}{27}\frac{(1-a^2l^2)^3}{l^2}=\frac{2}{27} \frac{(2\sqrt{3}al)^3}{l^2}=\frac{16}{9}a^3l \sqrt{3} \label{triple_root_M_3} \end{equation}

Comparing \eqref{triple_root_M_3} and \eqref{triple_root_M} we see that the condition $\Delta(P)=0$ is satisfied, which concludes the proof.
\end{proof}

We have now characterised all the cases where $P$ has a root with multiplicity $>1$, in the case of the double root we can also show:

\begin{prop}
If P has a root $x$ with multiplicity exactly 2 and \begin{equation} \displaystyle M^2=\frac{(1-a^2l^2)(a^4l^4+34a^2l^2+1)+ \varepsilon \sqrt{\delta} }{54l^2}, \quad \varepsilon \in \{-1,1\}\end{equation} then:

\begin{equation} x = \frac{12a^2l^2+(1-a^2l^2)(1-a^2l^2+\varepsilon\sqrt{\gamma})}{18m^2l^4}=\frac{12a^2l^2+(1-a^2l^2)(1-a^2l^2+\varepsilon\sqrt{\gamma})}{18Ml^2}  \label{x_expression}\end{equation}

Where $\gamma = (a^2l^2-1)^2-12a^2l^2 = (a^2l^2-2\sqrt{3}al-1)(a^2l^2+2\sqrt{3}al-1)$
\end{prop}
\begin{proof}
To find the expression of $x$, solve equation $(ii'')$ of \eqref{V2} for $x^2$, and then use equation $(ii''')$ of \eqref{V3} to find $x$. To decide which root to take for $x^2$, introduce $\varepsilon' \in \{-1,1\}$ in front of the radical in the expression for $x^2$ and then square the expression obtained for $x$. Injecting into this new expression those of $M^2$ and $x^2$, it is straightforward to obtain an expression for $\varepsilon\sqrt{\delta}$. After simplification we find that $\varepsilon \sqrt{\delta} = \varepsilon' \gamma \sqrt{\gamma}$. Hence, using the lemma below: $\varepsilon' =\varepsilon$.
\begin{lemme}
$\delta=\gamma^3$
\end{lemme}
\end{proof}
Using this result, we can study the relative position of the double root $x$ with respect to the other two roots; the above expression \eqref{x_expression} shows immediately that $x>0$. As before, the other roots are those of the polynomial:
\begin{equation} X^2+2x X -\frac{a^2}{l^2x^2} \end{equation}

As expected one of the roots ($x_{-})$ will be negative and the other positive, the positive root is given by:

\begin{equation} x_+= -x +\sqrt{x^2+\frac{a^2}{l^2x^2}} \end{equation}

We see that $x_+ > x$ if and only if $ \sqrt{x^2+\frac{a^2}{l^2x^2}} > 2x >0$. This holds if and only if: $$\frac{a^2}{l^2x^2} > 3x^2$$
Or, equivalently: $$x^4 < \frac{1}{3} \frac{a^2}{l^2}$$
As $x^4=m^2 x -\frac{a^2}{l^2}$, we deduce that:

\begin{equation} x_+ > x \Leftrightarrow x < \frac{4}{3}\frac{a^2}{M} \end{equation}

Note that $x=\frac{4}{3}\frac{a^2}{M}$ corresponds to the case where there is a triple root.

Rewriting \eqref{x_expression} we have:

\begin{equation} x = \frac{4}{3}\frac{a^2}{M} + \frac{\gamma +(1-a^2l^2)\varepsilon\sqrt{\gamma}}{18Ml^2}  \label{x_expression2} \end{equation}

So if $\varepsilon=1$ then $\frac{\gamma +(1-a^2l^2)\varepsilon\sqrt{\gamma}}{18Ml^2} > 0$ and so $x_+ < x $. In this case the outer black hole horizon has merged with the cosmological horizon.

If $\varepsilon = -1$ we show that $\frac{\gamma -(1-a^2l^2)\sqrt{\gamma}}{18Ml^2} < 0$ and so $x_+ > x$; the two black hole horizons have merged. This is the closest Kerr-de Sitter analog of extreme Kerr.

In order to show that: $\frac{\gamma -(1-a^2l^2)\sqrt{\gamma}}{18Ml^2} \leq 0$ we only need to study the sign of $\sqrt{\gamma} - (1-a^2l^2)$. i.e. the sign of:
$$ f(y)= \sqrt{(1-y^2)^2-12y^2}-(1-y^2)$$ when $0 \leq y \leq 2-\sqrt{3}$

But $ f(y)$ has same sign as : \begin{align*} f(y)(\sqrt{(1-y^2)^2-12y^2}+(1-y^2))&=(1-y^2)^2 -12y^2 -(1-y^2)^2\\&=-12y^2 <0 \end{align*}

To summarise, we have found three cases where horizons coincide: 

\begin{prop}
\label{extreme_kds}
Let $(a,l,M)\in \mathbb{R}_+^*$, then:

2 horizons coincide if and only if the both of the following conditions are satisfied:
\begin{itemize}
\item[$(i)$] $al< 2-\sqrt{3}$
\item[$(ii)$]$\displaystyle M^2=\frac{(1-a^2l^2)(a^4l^4+34a^2l^2+1) \pm \sqrt{\delta} }{54l^2}=m^2_{\pm}$
\end{itemize}
More precisely:
\begin{itemize}
\item If $M^2=m^2_+$ then the outer black hole horizon coincides with the the other cosmological horizon.
\item If $M^2= m^2_-$ then the two black hole horizons coincide.
\end{itemize}

Finally, if $al=2-\sqrt{3}$ and $M^2$ satisfies $(ii)$ then all three horizons situated in the region $r>0$ coincide.

\end{prop}

\subsection{Fast and slow Kerr-de Sitter}
We will now move on to study the Kerr-de Sitter equivalents to the usual so-called ``fast" and ``slow" Kerr black holes. Fast Kerr usually correspond to the case where there are no horizons. It owes its name to the fact that when $l=0$, it is completely characterised by the condition $a^2>M^2$. "Slow" Kerr, on the other hand, is characterised when $l=0$ by the condition $a^2<M^2$. In terms of the roots of the polynomial these cases correspond respectively, when $l=0$, to $\Delta_r$ having no roots, or $\Delta_r$ having two distinct real roots. As we have already noted, there are always two distinct roots with opposite sign in the case $l>0$ of Kerr-de-Sitter which correspond to the cosmological horizons inside and outside the singularity. Hence, in terms of roots the natural analogs for the Kerr-de Sitter metric are:

\begin{itemize}
\item $P$ has 4 distinct real roots (``Slow" Kerr-de Sitter)
\item $P$ has a complex root (``Fast" Kerr-de Sitter)
\end{itemize}

A further accommodating consequence of the necessary existence of two distinct real roots is that we can distinguish between the above cases using the sign of $\Delta(P)$. Indeed, let us denote the roots of $P$ by $x_1,x_2,x_3,x_4$ and assume, without loss of generality, that $x_1$ and $x_2$ are both real and distinct.

From proposition~\ref{discr_expression} of appendix~\ref{app_resultant} we can write (in $\mathbb{C}$):

$$\Delta(P)=(x_1-x_2)^2(x_1-x_3)^2(x_1-x_4)^2(x_2-x_3)^2(x_2-x_4)^2(x_3-x_4)^2$$
From this expression we see that if $x_3 \in \mathbb{R}$, $\Delta(P)\geq 0$\footnote{The discussion in the previous section shows that necessarily $x_4\in \mathbb{R}$ too}. If, however, $x_3=z \in \mathbb{C}\setminus \mathbb{R}$ then $x_4=\bar{z}$, hence:
\begin{align*} \Delta(P)&=(x_1-x_2)^2(x_1-z)^2(x_1-\bar{z})^2(x_2-z)^2(x_2-\bar{z})^2(2i \Im(z))^2 \\&=-4\Im(z)^2(x_1-x_2)^2 |x_1 -z|^2 |x_2-z|^2 <0 \end{align*}

Therefore, P has two conjugate complex roots if and only if $\Delta(P)<0$.

We recall the expression of $\Delta(P)$ of the previous section:

\begin{align} \Delta(P)= - \frac{16}{l^{10}}\left( 27M^4l^2 +(a^2l^2-1)(a^4l^4+34a^2l^2+1)M^2 +a^2(a^2l^2+1)^4 \right) \end{align}

The expression $27M^4l^2 +(a^2l^2-1)(a^4l^4+34a^2l^2+1)M^2 +a^2(a^2l^2+1)^4$ is a second order polynomial in $M^2$ whose discriminant is given by:
$$\delta = \gamma^3=(y-(2-\sqrt{3}))^3(y+2+\sqrt{3})^3(y+2-\sqrt{3})^3(y-(2+\sqrt{3}))^3$$

where $y=al$

From this factorisation we deduce the sign of $\delta$ given in table~\ref{sign_delta}, and the following cases:

\begin{enumerate}
\item[$(i)$] $0 \leq al \leq 2-\sqrt{3}$: \par
In this case $\Delta(P)= -\frac{432}{l^8}(M^2-m^2_-)(M^2-m^2_+)$ where $0 \leq m^2_- \leq m^2_+$. It follows that if $M^2\in [m^2_-,m^2_+]$ then $\Delta(P)\geq0$ otherwise, $\Delta(P)<0$
\item[$(ii)$] $2-\sqrt{3} < al < 2+\sqrt{3}$: \par Here $\Delta(P)$ never vanishes for any value of $M^2$. Since for $M^2=0$, $\Delta(P)<0$ and $\Delta(P)$ is a continuous function of $M^2$, $\Delta(P)<0$ for all values of $M^2$.
\item[$(iii)$] $al \geq 2+\sqrt{3}$: \par $\Delta(P)=-\frac{432}{l^8}(M^2+m^2_-)(M^2+m^2_+)$ where $0 \leq m^2_+ \leq m^2_-$ Therefore, for all values of $M\geq 0$, $\Delta(P)<0$.
\end{enumerate}
\begin{table}[h]

\centering
\caption{Sign of $\delta$}
\begin{tikzpicture}
\label{sign_delta}
   \tkzTabInit{$y=al$ / 1 , Sign of $\delta$ / 1}{$0$, $2-\sqrt{3}$, $2+\sqrt{3}$, $+\infty$}
   \tkzTabLine{ ,+, z, -, z, +, }
\end{tikzpicture}
\end{table}

Combined with the results of the previous section and preserving the terminology introduced at the beginning of this section, we have thus shown:

\begin{prop}
\label{prop:carac_slow_fast}
\begin{itemize}
\item ``Slow" Kerr de Sitter is characterised by the following conditions on the parameters $(a,l,M)\in\mathbb{R}^*_+$

\begin{enumerate}
\item[$(i)$] $al<2-\sqrt{3}$
\item[$(ii)$] $M^2\in[m^2_-,m^2_+]$ where $m^2_{\pm}=\frac{(1-a^2l^2)(a^4l^4+34a^2l^2+1) \pm \sqrt{\delta} }{54l^2}$
\end{enumerate}

\item ``Fast" Kerr-de Sitter corresponds to the cases:

\begin{enumerate}

\item[$\triangleright$] $0\leq al \leq 2-\sqrt{3}$ and $M^2 \not\in [m^2_-,m^2_+]$ where $m^2_{\pm}=\frac{(1-a^2l^2)(a^4l^4+34a^2l^2+1) \pm \sqrt{\delta} }{54l^2}$

This is the case that most ressembles the usual fast Kerr spacetime. 

\item[$\triangleright$] $al >2-\sqrt{3}  $
\end{enumerate}
\end{itemize} 
\end{prop}

In the above proposition we see the black hole horizons exist on a de Sitter background only under relatively strict conditions on the parameters, we have notably, for a given value of $\Lambda$, upper \emph{and} lower bounds on the mass, as well as a restriction on the spin $a$ of the black hole. Let us concentrate for a moment on the upper bound for the mass for a given values of $a,l, al<2-\sqrt{3}$ of a slow KdS spacetime. According to condition $(ii)$, we must have: \begin{equation} M^2\leq \frac{(1-a^2l^2)(a^4l^4+34a^2l^2+1) + \sqrt{\delta}}{54l^2} \end{equation}

 Despite our assumption that $a>0$, setting $a=0$ and taking the square root furnishes a well known result in Schwarzschild-de Sitter spacetime~\cite{stuchlik:1999aa}: \begin{equation} M<\frac{1}{3\sqrt{\Lambda}} \end{equation}
 More generally, the map $y\mapsto (1-y^2)(y^4+34y^2+1) +\sqrt{\delta(y)}$, is well defined and continuous for $y\in [0,2-\sqrt{3}]$ and attains a maximum at $y=2-\sqrt{3}$. This yields a global bound on the mass: $M< \frac{C}{\sqrt{\Lambda}}$ where $C=\frac{4}{\sqrt{3}}\sqrt{26\sqrt{3}-45} \approx 0.4215$. Studying how the expression of the upper bound depends on $a$, it can be shown that in fact the minimum value is attained for $a=0$: rotating black holes can be slightly more massive than non-rotating black holes and still maintain their horizon structure.

We conclude this section by addressing one last question regarding slow Kerr-de Sitter black hole: can there be more one than one horizon inside the singularity, i.e. in the region $r<0$? The answer is no, as shown in the following lemma.
\begin{lemme}
We suppose $a\neq 0$. In slow Kerr-de Sitter only one horizon lies in the region $r<0$
\end{lemme}

\begin{proof}
It has already been noted that there must always be at least one negative root; an even number of both positive and negative roots is excluded again by equation $(iv)$ in \eqref{Viete}. The statement of the lemma is therefore equivalent to the fact that there cannot be 3 negative roots. As usual, denote by $x_1,x_2,x_3,x_4$ the 4 roots of $\Delta_r$. By hypothesis, they are all real. Suppose, without loss of generality, $x_1x_2<0$. It follows that $x_3x_4>0$ from equation $(iv)$ of \eqref{Viete}. Call $P=x_3x_4$ and $S=x_3+x_4$. Equation $(i)$ of \eqref{Viete} gives: $S=-(x_1+x_2)$. Equation $(iii)$ of \eqref{Viete} yields:

$$ -\frac{A^2}{P}S -SP=-2m^2 $$
Which is equivalent to: 
$$S=\frac{2m^2P}{A^2+P^2}\geq 0$$
Therefore $S=x_3+x_4$ is always positive and thus $x_3$ and $x_4$ are both positive.
\end{proof}

\subsection{Boyer-Lindquist blocks}
\label{section:bl_blocks}
We are now in a position to give a more precise description of the Boyer-Lindquist blocks. We will do this first in the slow case, where there are four distinct roots, say, $r_{--},r_{-},r_{+},r_{++}$ ordered as:
$$ r_{--}< 0 < r_{-} \leq r_{+} \leq r_{++} $$

In table~\ref{sign_delta_r} we give the sign of $\Delta_r$ as $r$ varies and the chosen numbering for the Boyer-Lindquist blocks. We also give the sign of the diagonal metric tensor elements $g_{ii}$. The ``$\bullet$" means that the sign changes within the block. That $g_{\phi\phi}>0$ for $r>0$ is not clear from the initial expression of $g_{\phi\phi}$ given in table~\ref{BLmetric}, however one can write:

\begin{equation} g_{\phi\phi}=\left((r^2+a^2)+\frac{2Mra^2\sin^2\theta}{\rho^2 \Xi} \right)\frac{\sin^2\theta}{\Xi} \end{equation}

\begin{table}[h]

\centering

\begin{tikzpicture}
   \tkzTabInit[espcl=2]{$r$ / 1 , $\Delta_r$ / 1, Boyer Lindquist blocks /2, $g_{tt}$ /1 , $g_{rr}$ /1, $g_{\theta\theta}$ /1, $g_{\phi\phi}$ /1, $g(V\text{,}V)$ /1, $g(W\text{,}W)$ /1}{$-\infty$, $r_{--}$,0, $r_{-}$, $r_{+}$, $r_{++}$ , $+\infty$}
   \tkzTabLine{,-,z,,+,,z,-,z,+,z,-}
   \tkzTabLine{,\text{V},d,,\text{IV},,d,\text{III},d,\text{II},d,\text{I}}
   \tkzTabLine{,+,d,,\bullet,,d,+,d,\bullet,d,+}
   \tkzTabLine{,-,d,,+,,d,-,d,+,d,-}
   \tkzTabLine{,+,d,,+,,d,+,d,+,d,+}
   \tkzTabLine{,-,d,\bullet,t,+,d,+,d,+,d,+}
   \tkzTabLine{,+,d,,-,,d,+,d,-,d,+}
   \tkzTabLine{,+,d,,+,,d,+,d,+,d,+}
\end{tikzpicture}
\caption{Sign of $\Delta_r$ and Boyer-Lindquist blocks \label{sign_delta_r}}
\end{table}

Up to now, we have not addressed the question of the time-orientation\footnote{A time orientation of a Lorentzian manifold is a choice of a globally defined nowhere vanishing non-spacelike continuous vector field. A vector field is said to be \emph{time-orientable} if such a vector field exists} of the manifolds under consideration. The time-orientability of each Boyer-Lindquist block is clear from table~\ref{sign_delta_r}, so each Boyer-Lindquist block can separately become a \emph{spacetime}. For the usual Kerr metric and the Schwarzchild metric, the time parameter $t$ coincides with the proper time of a distant stationary observer in the limit $r\to \infty$. In this case, time-orientation of the Boyer-Lindquist block that lies beyond all black hole horizons can be chosen naturally under the prescription that $\partial_t$ is future-pointing when non-space-like. This interpretation of $t$ fails for the Kerr-de Sitter metric, but we still have a number of partial results. First, under the assumption that our visible universe is not beyond a cosmological horizon and not between two black hole horizons, block II (cf table~\ref{sign_delta_r}) is identified as the most physically relevant block. On this block $t$ is still a ``time function'' in the following sense:
\begin{lemme}
\label{lemme:fonc_temps}
On block II, the hypersurfaces ``$t=t_0$'' are spacelike.
\end{lemme}
\begin{proof} 
At each point $p$ of such a surface the tangent space is given by the kernel of $dt_p$, or, equivalently $(\nabla t(p))^{\perp}$. But, $\nabla t$ is timelike on block II ( minus axes ) since\footnote{Refer to lemma~\ref{lemme:ginverse},\ref{lemme:grad_t} in appendix~\ref{app:diverse}} $g(\nabla_t,\nabla_t)=g^{tt}=-\frac{g_{\phi \phi} \Xi^4}{\sin^2\theta \Delta_\theta \Delta_r}$. This also holds for points on the axes, as this expression extends continuously to such points.\end{proof}
\begin{corollaire}
Along any non-spacelike $C^1$ curve $\alpha$ in block II, $t\circ \alpha$ is strictly monotonic.
\end{corollaire}
The region in the Kerr-Boyer-Lindquist blocks where $g_{tt} >0$ is known as the ``ergosphere''. It has interesting physical properties explored in~\cite{ONeill:2014aa} in the Kerr case, the most notable of which being the possibility to extract energy from a Kerr black hole. In the case of the Kerr-de Sitter metric it is no longer guaranteed that the ergosphere does not cover all of block II, unless we impose further conditions:
\begin{prop}
Suppose $a^2l^2<1$, then a sufficient condition for there to be an interval $I \subset \mathbb{R}^*_+$ such that $g_{tt}\leq 0$ when $r\in I$ is that \begin{equation} 27M^2l^2\leq (1-a^2l^2)^3\end{equation}
\end{prop}
\begin{proof}
Rewrite $g_{tt}$ as:

$$ g_{tt} = \frac{1}{\rho^2 \Xi^2} \left( \underbrace{a^2\cos^2\theta( l^2a^2\sin^2\theta -1)}_{\leq 0} +l^2r\left(r^3 + r\frac{(a^2l^2-1)}{l^2}+\frac{2M}{l^2}\right)\right)$$

$a^2l^2<1$, hence $l^2a^2\sin^2\theta \leq 1$, so the first term is always non-positive. The sign of the second term is determined by that of the polynomial:
$$ P = X^3 + X\frac{a^2l^2-1}{l^2} +\frac{2M}{l^2}$$

It can become negative on $\mathbb{R}^*_+$ if and only if there is a positive real root, hence its discriminant must be positive. This is because if there is only one real root, it must be negative as $\frac{2M}{l^2}>0$. The discriminant of P is given by:

$$\Delta(P)=(1-a^2l^2)^3-27M^2l^2$$

It is positive if and only if $27M^2l^2 \leq (1-a^2l^2)^3$ and in this case all roots are real, but they cannot all be negative since their sum must vanish.

\end{proof}

$t$ is nevertheless a ``function of time'' and, even though there are cases where $\partial_t$ is always space-like, its gradient always furnishes on block II a time-like vector field that can be used to time-orient it. By analogy with the Kerr case, we choose to time-orient block II by specifying that $-\nabla t$ is future-pointing.


\section{Maximal Kerr-de Sitter spacetimes}
\label{maximal}
In this section we will cease to consider the Boyer-Lindquist blocks as separate spacetimes and construct analytical manifolds containing isometric copies of these blocks, of which the union is dense, and to which the Kerr-de Sitter metric extends analytically. In order for these manifolds to be spacetimes they will be constructed in such a way to ensure that they are time-orientable.
The methods used here are adapted from~\cite{ONeill:2014aa} and are still applicable due to the remarkable algebraic decomposition of the Riemann curvature tensor described in section~\ref{section:kds_metric}. 
\subsection{$KdS^*$ et ${}^*KdS$ spacetimes}
The first two analytical manifolds will be constructed by choosing coordinates for the Boyer-Lindquist blocks in which one of the two null geodesic congruences generated by the vector fields \begin{equation} N_{\pm} = \pm \partial_r +\frac{\Xi}{\Delta_r}V \end{equation} are coordinate-lines.
Recall from proposition~\ref{prop:kds_conformal_properties} that at each point $p\in \mathcal{B}$ of any Boyer-Lindquist block $\mathcal{B}$ the rays generated by the vectors $N_{\pm}(p)$ define the principal null directions. The geometric significance of these directions justifies using them to construct an analytical extension.

\begin{definition}
\label{def:starcoordinates}
We define $KdS^*$ coordinates by:

\begin{equation} \left\{\begin{array}{c} t^* = t + T(r) \\ r^*=r \\ \theta^* = \theta \\ \phi^*=\phi +A(r) \end{array}\right. \end{equation}

Similarly ${}^*KdS$ coordinates are defined by:
\begin{equation} \left\{\begin{array}{c} {}^*t = t - T(r) \\ {}^*r=r \\ {}^*\theta = \theta \\ {}^*\phi=\phi - A(r) \end{array}\right. \end{equation}

Where $\displaystyle T(r)= \int \frac{(r^2+a^2)\Xi}{\Delta_r} \dd r$ and $ \displaystyle A(r)= \int \frac{a \Xi}{\Delta_r} \dd r$
\end{definition}

\subsection{$KdS^*$}
\label{kds_star}
\begin{prop}
Let $\mathcal{B}$ be a Boyer-Lindquist block and $\mathcal{A}=\mathbb{R}_t\times \mathbb{R}_r\times \{ p_{\pm}\}$; $p_{\pm}$ denote the poles of the $S^2$. Define:
$ \Phi^*: \mathcal{B}\setminus {\mathcal{A}} \longrightarrow \mathbb{R}_{t^*}\times \mathbb{R}_{r^*}\times S^2 $ by:
$\Phi^*(t,r,\theta,\phi) = ( t+T(r),r,\theta, \phi + A(r)) $
then $\Phi^*$ is an analytic diffeomorphism of $\mathcal{B}\setminus \mathcal{A}$ onto an open subset of $\mathbb{R}_{t^*}\times \mathbb{R}_{r^*}\times S^2$
\end{prop}
\begin{proof}
That $\Phi^*$ is analytic is clear; fix $(t,r,\theta,\phi)\in\mathcal{B}\setminus \mathcal{A}$, then the Jacobian matrix is given by:

$$J(\phi)(t,r,\theta,\phi)= \left(\begin{array}{cccc} 1 & \frac{r^2+a^2}{\Delta_r}\Xi & 0 & 0 \\ 0 & 1 & 0 & 0 \\ 0 & 0 & 1 & 0 \\ 0 & \frac{a\Xi}{\Delta_r} & 0 & 1  \end{array}\right) $$

Thus, $\det J(\phi)(t,r,\theta,\phi)=1$. It follows that $\Phi^*$ is a local analytic diffeomorphism at each point of $\mathcal{B} \setminus \mathcal{A}$. It suffices to show that $\Phi^*$ is injective to conclude that it is a global diffeomorphism. Injectivity is clear however, as, by definition~\ref{def:starcoordinates}: $$ \Phi^*(r,t,\theta,\phi)= \Phi^*(r',t',\theta',\phi') \Leftrightarrow \left\{ \begin{array}{c} t + T(r) = t' +T(r') \\ r= r' \\ \theta=\theta' \\ \phi + A(r)= \phi' + A(r') \end{array}\right. \Leftrightarrow \left\{ \begin{array}{c} t = t'  \\ r= r' \\ \theta=\theta' \\ \phi= \phi'\end{array}\right.  $$
\end{proof}

$(t^*,r,\theta,\phi^*)$ are therefore coordinates functions on $\mathcal{B} \setminus \mathcal{A}$
\begin{lemme}
The coordinate vector fields $\partial_{t^*}, \partial_{r^*}, \partial_{\theta^*}, \partial_{\phi^*}$ are  given on each Boyer-Lindquist block by:

\begin{equation}\label{eq:kstar_coord_fields} \partial_{t^*} = \partial_t \quad \quad  \partial_{r^*}= \partial_r - \frac{\Xi}{\Delta_r}V=-N_- \quad \quad  \partial_{\theta^*}=\partial_{\theta} \quad\quad \partial_{\phi^*}=\partial_{\phi} \end{equation}

Furthermore, in $KdS^*$ coordinates the line element can be written:

\begin{equation}\label{eq:kstar_line_element} \dd s^2 = g_{tt}\dd{t^*}^2 + g_{\theta\theta}\dd{\theta^*}^2 + g_{\phi\phi}\dd{\phi^*}^2 +\frac{2}{\Xi}\dd{t^*}\dd{r^*} -\frac{2a\sin^2\theta}{\Xi}\dd{r^*}\dd{\phi^*} +2g_{\phi t}\dd{t^*}\dd{\phi^*}  \end{equation}
\end{lemme}

\begin{corollaire}
On each Boyer-Lindquist block $\mathcal{B}$ the integral curves of $N_{-}$ are the coordinate lines of $r^*$.
\end{corollaire}
Inspecting the form of~\eqref{eq:kstar_line_element} and comparing with the discussion at the beginning of section~\ref{section:kds_metric} we deduce:
 \begin{corollaire}
 By analogy with the notations used in section~\ref{section:kds_metric}, let $\Sigma^*=\{(t^*,r^*,\theta^{*},\phi^*)\in \mathbb{R}_{t^*}\times\mathbb{R}_{r^*}\times S^2, {r^*}^2+a^2\cos^2\theta^*=0\}$, then the line element~\eqref{eq:kstar_line_element} extends analytically to all of $\mathbb{R}_{t^*} \times \mathbb{R}_{r^*}\times S^2\setminus{\Sigma^*}$ as a non-degenerate metric tensor. 
 \end{corollaire}
 
 This last result leads us to define:
 
 \begin{definition}
 We call $KdS^*$ the analytical manifold $\mathbb{R}_{t^*}\times \mathbb{R}_{r^*} \times S^2 \setminus \Sigma^*$ equipped with metric tensor $g^*$ defined by~\eqref{eq:kstar_line_element} and time-oriented such that $-\partial_{r^*}$ is future-pointing.
 \end{definition}
 
 \begin{rem}
 \begin{itemize}
 \item Time-orientation is chosen here so that the integral curves (and coordinate lines) of $N_-$ are future-oriented
 \item It is consistent with the choice that $-\nabla t$ is future-pointing on block II, since, using \eqref{eq:kstar_line_element} and lemma~\ref{lemme:grad_t} in appendix~\ref{app:diverse}, it is easily seen that $g^*(-\partial_{r^*},-\nabla t)=g^*(\partial_{r^*}, \nabla t)=-\frac{\Xi}{\Delta_r}(r^2+a^2)<0$
 \end{itemize}
 \end{rem}
 
Define now the subsets $\mathcal{B}^*$ of $KdS^*$  by the same inequalities as the corresponding Boyer-Lindquist blocks $\mathcal{B}$

\begin{lemme}
$\Phi^*$ has an analytic extension to a diffeomorphism of $\mathcal{B}$ onto $\mathcal{B}^*$
\end{lemme}
\begin{proof}
For $\alpha\in \mathbb{R}$, let $R_\alpha: S^2 \longrightarrow S^2$ be the restriction of the rotation of angle $\alpha$ about the $z$-axis in $\mathbb{R}^3$ to $S^2$. The map $\psi: \mathbb{R}_r \times S^2 \longrightarrow S^2$ defined by $\psi(r,q)=R_{A(r)}(q)$ is analytic everywhere except at values of $r$ where $\Delta_r=0$. Then: 

\begin{equation*} \func{\tilde{\Phi}^*}{\mathcal{B}}{\mathcal{B}^*}{(t,r,q\in S^2)}{(t+T(r),r,\psi(r,q))} \end{equation*}
is the desired extension.
\end{proof}

\begin{corollaire}
Each Boyer-Lindquist block $\mathcal{B}$ can be identified isometrically with an open subset of $KdS^*$.
\end{corollaire}

The vector fields $\partial_t, \partial_\theta, \partial_\phi$ are, a priori, only well defined on each $\mathcal{B}^*$, but, in view of equation~\eqref{eq:kstar_coord_fields}, $\partial_{t^*}, \partial_{\theta^*}, \partial_{\phi^*}$ are analytic extensions of these fields to all of $KdS^*$. Hence, we define $\partial_t, \partial_\theta$ and $\partial_\phi$ by equation~\eqref{eq:kstar_coord_fields} on all of $KdS^*$.

The hypersurfaces $\mathscr{H}^*_i$ defined by the equations $r=r^*=r_i$ ($i\in\{--,-,+,++\}$) are now well-defined submanifolds of $KdS^*$, it is easy to show that, as is custom with black hole horizons:
\begin{prop}
Each $\mathscr{H}^*_i$ is a totally geodesic null hypersurface of $KdS^*$.

In particular, for $p\in \mathscr{H}^*_i$:

$$T_p\mathscr{H}^*_i = V_p^{\perp}=\textrm{span}\left( (\partial_t)_p,(\partial_\theta)_p,(\partial_\phi)_p) \right)=\textrm{span}\left( V_p, (\partial_\theta)_p, (\partial_\phi)_p \right)$$
\end{prop}

We shall now address the question of the integral curves of $N_+$ in $KdS^*$, the situation is not symmetrical with that of $N_-$, as, in terms of the $KdS^*$ coordinate fields:

$$N_+= \partial_{r^*} + \frac{2\Xi}{\Delta_r}V$$

Thus, $N_+$ is still undefined on the horizons $\mathscr{H}_i$, moreover, $N_+$ is not always future-pointing since:
$$g^*(N_+,-\partial_{r^*}) = -\frac{2\rho^2}{\Delta_r}$$

However this can be remedied by considering reparametrisations of the integral curves of $N_+$ that are integral curves of $n_+=\frac{\Delta_r}{2\Xi}N_+$. The integral curves of $n_+$ are all future-oriented since $g^*(n_+,-\partial_{r^*})=\frac{-\rho^2}{\Xi^2} < 0$

\begin{definition}
On $KdS^*$ we will call:

\begin{enumerate}
\item ``Ingoing principal null geodesics" the integral curves of the vector field $N_-$ extended to all of $KdS^*$ by~\eqref{eq:kstar_coord_fields}
\item ``Outgoing principal null geodesics" geodesic reparametrisations of the integral curves of $n_+$.

These curves coincide on $\mathcal{B}^*$ with the images of the principal null geodesics of the Boyer-Lindquist blocks by $\tilde{\Phi}^*\equiv i^*$.
\end{enumerate}

\end{definition}

In figure~\ref{img:kstar}, we give a schematic representation of $KdS^*$ spacetime that will be useful in the following. The principal null geodesics are represented by oriented line segments; horizontally, the ``ingoing" principal null geodesics run from $r=+\infty$ to $r=-\infty$ - we will say that they are ``complete" -, vertically, the ``outgoing" principal null geodesics are confined within a given Boyer-Lindquist block. We have not represented the principal null geodesics that are confined within the horizons.
\begin{figure}[h]
\centering
\includegraphics[scale=.45]{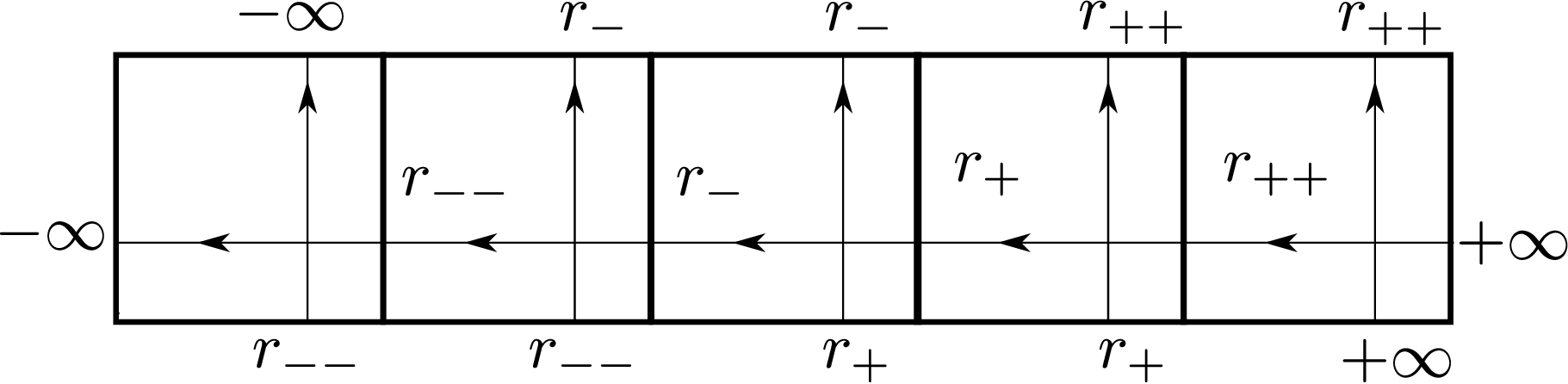}
\caption{Schematic representation of $KdS^*$ spacetime: horizontally, the ingoing principal null geodesics run unimpeded from $r=+\infty$ to $r=-\infty$, vertically, the outgoing principal null geodesics are confined within a given Boyer-Lindquist block and on the horizons.\label{img:kstar}}
\end{figure}

\subsection{${}^*KdS$}
\label{section:starkds}
Repeating the above arguments, using instead ${}^*KdS$ coordinates, yields the following results:

\begin{lemme}
\begin{enumerate}
\item On each Boyer-Lindquist block $({}^*t,{}^*r,{}^*\theta, {}^*\phi)$ are well defined coordinate functions.
\item In these coordinates the line element can be written:
\begin{equation} \label{eq:stark_line_element} \dd s^2 = g_{tt}\dd{{}^*t}^2 + g_{\theta\theta}\dd{\kstar\theta}^2 + g_{\phi\phi}\dd{\kstar\phi}^2 -\frac{2}{\Xi}\dd{\kstar t}\dd{\kstar r} +\frac{2a\sin^2\theta}{\Xi}\dd{\kstar r}\dd{\kstar\phi} +2g_{\phi t}\dd{\kstar t}\dd{\kstar \phi} \end{equation}
This expression has an unique analytic extension to all points of $\mathbb{R}_{\kstar t}\times \mathbb{R}_{\kstar r} \times S^2 \setminus  \kstar \Sigma$ 
\item The coordinate vector fields are:
\begin{equation}\label{eq:stark_coord_fields} \partial_{\kstar r} = \partial_r + \frac{\Xi}{\Delta_r}V = N_+ \quad\quad \partial_{\kstar t} = \partial_t \quad\quad \partial_{\kstar \theta} =\partial_\theta \quad \quad \partial_{\kstar \phi}=\partial_\phi \end{equation}
\end{enumerate}
\end{lemme}
\begin{prop}
Define the Lorentizan manifold $\kstar KdS$ to be the analytic manifold $\mathbb{R}_{\kstar t} \times \mathbb{R}_{\kstar r} \times S^2 \setminus \kstar\Sigma$ equipped with the metric $^*g$ defined by equation~\eqref{eq:stark_line_element} and time-oriented such that the globally defined vector field $\partial_{\kstar r}$ is future-pointing then:

\begin{enumerate}
\item The submanifolds $\kstar \mathscr{H}_i$ of equations $r=r_i, i\in\{--,-,+,++\}$ are totally geodesic null hypersurfaces.
\item Defining $\kstar\mathcal{B}$ by the same inequalities as the Boyer-Lindquist block $\mathcal{B}$, then $\kstar\mathcal{B}$ and $\mathcal{B}$ are isometric, i.e. $\kstar KdS$ contains isometric copies of each Boyer-Lindquist block.
\end{enumerate}
\end{prop}

\begin{definition}
On $\kstar KdS$ we will call:

\begin{enumerate}
\item ``Outgoing principal null geodesics" the integral curves of the vector field $N_+$ extended to all of $\kstar KdS$ by~\eqref{eq:stark_coord_fields}.
\item ``Ingoing principal null geodesics" geodesic reparametrisations of the integral curves of the everywhere future-pointing vector field $n_{-}= \frac{\Delta_r}{2\Xi}{N_-}$
\end{enumerate}
\end{definition}

\begin{figure}[h]
\centering
\includegraphics[scale=.45]{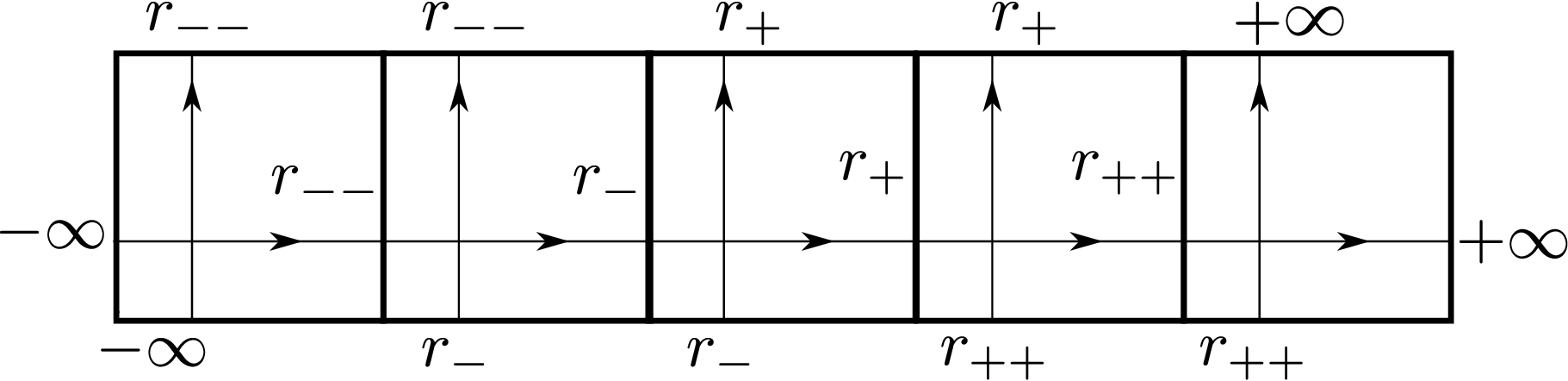}
\caption{Schematic representation of $\kstar KdS$ spacetime\label{img:stark}}
\end{figure}

In figure~\ref{img:stark}, we give the corresponding schematic representation of $\kstar KdS$. Again, the principal null geodesics are represented by oriented line segments. Here though, horizontally, are the \emph{outgoing} principal null geodesics running from $r=-\infty$ to $r=+\infty$ and vertically, the \emph{ingoing} principal null geodesics confined within a single Boyer-Lindquist block $\kstar \mathcal{B}$. Again, we have omitted the ingoing principal null geodesics trapped in the horizon.

The asymmetric treatment of the outgoing and ingoing principal null geodesics shows that $\kstar KdS$ and $\kstar KdS$ are certainly not the same spacetime. Nevertheless, there is a natural isometry $\mu$ between $\kstar \mathcal{B}$ and $\mathcal{B}^*$ for each Boyer-Lindquist block $\mathcal{B}$, in coordinates it can be written: 

\begin{equation}\label{eq:canonical_isometry}\mu(\kstar t, \kstar r, \kstar \theta, \kstar \phi) = (\kstar t + 2 T(r), \kstar r, \kstar \theta, \kstar \phi +2 A(r) ) \end{equation}

From which we deduce that:

\begin{equation*} \dd \mu (\partial_{\kstar r}) = \partial_r^* + \frac{2\Xi}{\Delta_r}V \end{equation*}

Hence: $$g^*(-\partial_{r^*},\dd\mu(\partial_{\kstar r })) = -\frac{2\rho^2}{\Delta_r}$$

Therefore, $\mu$ preserves time-orientation on blocks II and IV (see table~\ref{sign_delta_r}) but reverses it on blocks I, III and V.

We conclude this section defining two more spacetimes:
\begin{definition}
We define ${KdS^*}'$ and ${{}^*KdS}'$ to be the spacetimes obtained from ${KdS^*}$ and $\kstar{KdS}$ respectively by reversing time orientation.
\end{definition}

\begin{lemme}
For each Boyer-Lindquist block $\mathcal{B}$, the isometries $\kstar \mathcal{B} \longrightarrow {\mathcal{B}^*}'$ and $\kstar \mathcal{B}' \longrightarrow \mathcal{B}^*$ defined in coordinates by~\eqref{eq:canonical_isometry} preserve time-orientation on blocks I, III and V, but reverse it on blocks II and IV.
\end{lemme}

After reversing time-orientation, the principal null geodesics are now past-oriented. Their orientation should be reversed so that they are future-oriented, but because this changes the sign in front of $\partial_r$ in the original expression, we also adapt terminology: an orientation reversed integral curve of $\partial_{r^*}$ (resp. $\partial_{\kstar r}$)  will become an outgoing principal null geodesics in ${KdS^*}'$ (resp. $\kstar KdS'$) and similarly for the integral curves of $n_\pm$. The reason for this is purely semantic, in the next section we will seek to extend the incomplete outgoing principal null geodesics by gluing together along the Boyer-Lindquist blocks combinations of the four manifolds of this section, the change of vocabulary ensures that we always extend outgoing principal null geodesics using outgoing principal null geodesics.

\subsection{Maximal slow Kerr-de Sitter spacetime}
In the previous section we constructed four isometric - but not identical - analytic extensions of the KdS-Boyer-Lindquist blocks. In one case, ingoing principal null geodesics are complete, and in the other outgoing principal null geodesics are complete. In this section, we seek an analytical extension of these spacetimes such that all principal null geodesics, save those that run into the singularity, are complete, i.e. a maximal extension of these curves is defined on all of $\mathbb{R}$.
As for Kerr spacetime in~\cite{ONeill:2014aa}, the maximal extensions by ``gluing" together the aforementioned manifolds in an elaborate fashion. 

By ``gluing'' two semi-Riemannian manifolds $X$ and $Y$, we mean that we construct a new manifold $Q$ containing isometric copies of $X$ and $Y$ and equipped with a metric extending that of both $X$ and $Y$. A natural way of doing this is to specify two open sets $U\subset X$ and $V \subset Y$ that are identified by an isometry $\phi: U \longrightarrow V$, in this case we denote the new manifold by $X\coprod_{\phi}Y$. It comes with two ``canonical" embeddings $\bar{i}: X \longrightarrow Q, \bar{j}: Y \longrightarrow Q$ and $\bar{i}(X)\cap \bar{j}(Y) = \bar{i}(U) = \bar{j}(V)$.   A brief outline of the construction is given in appendix~\ref{app:gluing}, however we note here that whilst most topological properties of the new space $Q$ follow directly from those of $X$ and $Y$, separation is not guaranteed. Nevertheless, we have a technical criterion- proved in appendix~\ref{app:gluing} - that will suffice for all cases encountered in the sequel:

\begin{lemme} 
\label{lemme:separation2}
If $X$ and $Y$ are two manifolds and there is no sequence $(x_n)_{n\in\mathbb{N}}$ of points in $U$ converging to a point in $\bar{U}\setminus U$ and such that $\phi(x_n)_{n\in \mathbb{N}}$ converges to a point in $\bar{V}\setminus V$, then $Q$ is Hausdorff.
\end{lemme}

Throughout this section, we assume that the conditions of slow $KdS$ as described in section~\ref{delta_r} are satisfied. In particular, we assume that $\Delta_r$ has four distinct roots. Whilst some of the more technical results in this section are independent of this hypothesis, the gluing pattern is dependent of this choice.
 
\subsubsection{Kruskal domains}
\label{slow_kruskal}
Rather than directly gluing the manifolds $KdS^*,{}^*KdS$ and their orientation reversed counterparts, the pattern is more conveniently described by first constructing smaller manifolds, called ``Kruskal domains'', from selected open sets of these manifolds. Four such domains are required, one per horizon; they are illustrated in figure~\ref{figure:kruskal_slow} and are destined to be assembled by gluing along Boyer-Lindquist blocks sharing identical labels. Unprimed labels indicate that the blocks are time-oriented according to $KdS^*$, primed labels are worn by blocks with the opposite time-orientation.

\begin{figure}[h]
\centering
\begin{subfigure}{.4\textwidth}
\includegraphics[scale=.3]{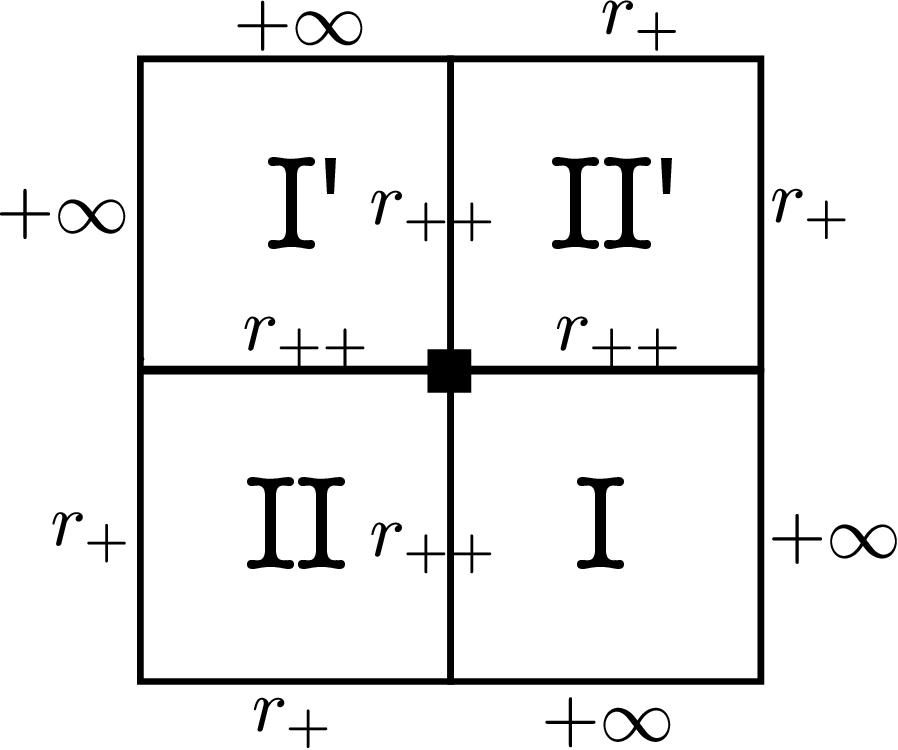}
\caption{$\mathscr{D}(r_{++})$}
\end{subfigure}
\begin{subfigure}{0.4\textwidth}
\includegraphics[scale=.3]{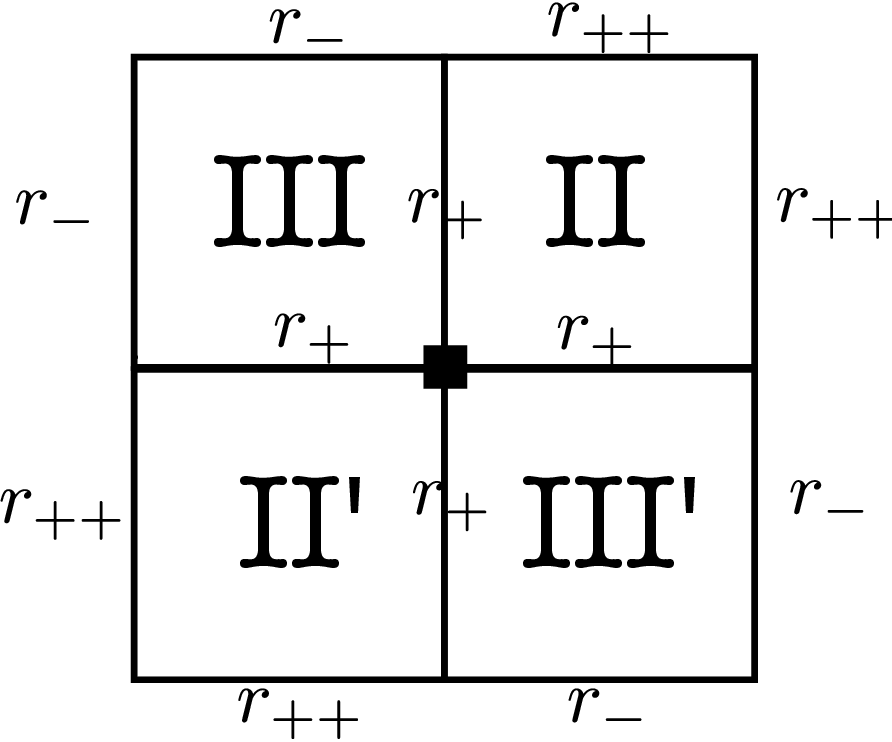}
\caption{$\mathscr{D}(r_{+})$}
\end{subfigure}

\begin{subfigure}{0.4\textwidth}
\includegraphics[scale=.3]{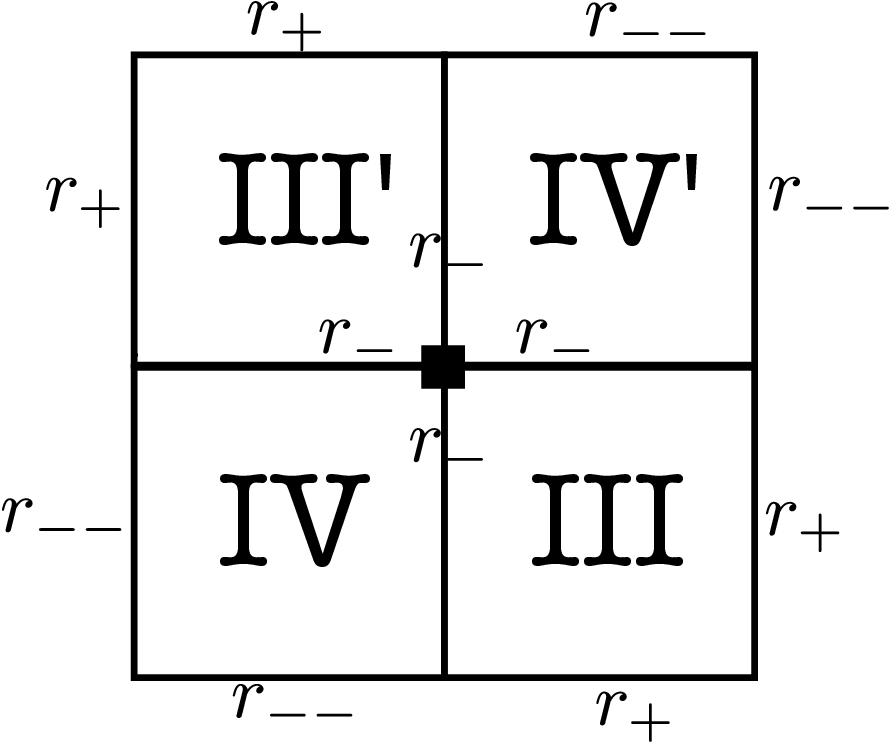}
\caption{$\mathscr{D}(r_-)$}
\end{subfigure}
\begin{subfigure}{0.4\textwidth}
\includegraphics[scale=.3]{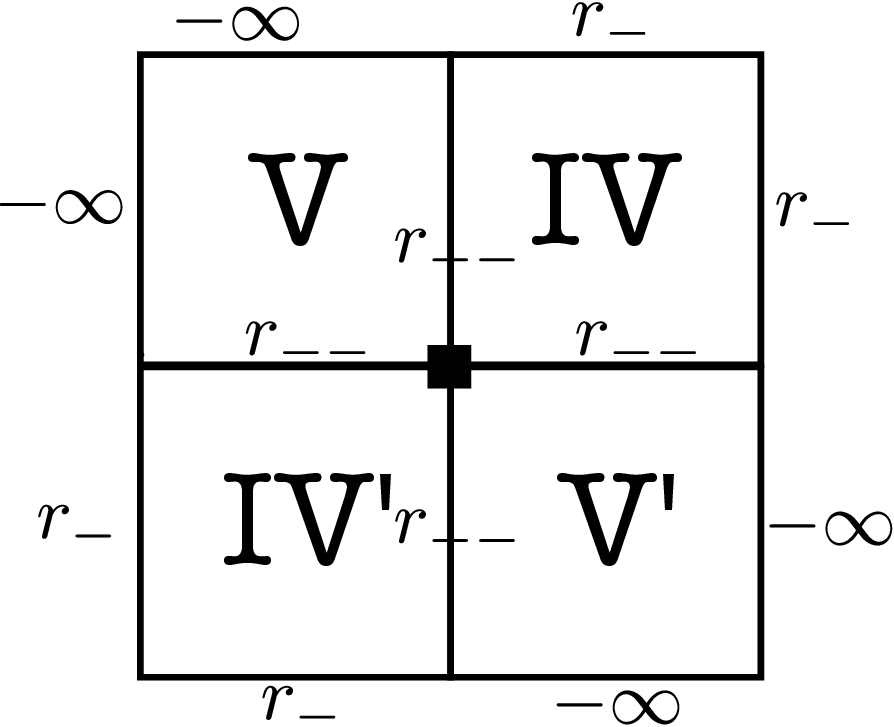}
\caption{$\mathscr{D}(r_{--})$}
\end{subfigure}

\caption{\label{figure:kruskal_slow} Kruskal domains, the black square is the crossing-sphere (see section~\ref{section:crossing_sphere})}
\end{figure}
The Kruskal domains are also built in two stages. First, chosen open sets - that contain selected Boyer-Lindquist blocks - are glued together using the isometries discussed at the end of section~\ref{section:starkds}; the result of this will be a manifold $\mathscr{D}_0(r_i)$. However, closer analysis of the principal null geodesics contained within the horizons of $KdS^*$ and ${}^*KdS$ will show that $\mathscr{D}_0(r_i)$ does not complete all principal null geodesics as required and will also need to be extended. Let us consider, as an example, $\mathscr{D}_0(r_{++})$; the other domains can be constructed similarly. 

$\mathscr{D}_0(r_{++})$ is built according to figure~\ref{figure:construct_dpp}.
\begin{figure}[h]
\centering
\includegraphics[scale=.3]{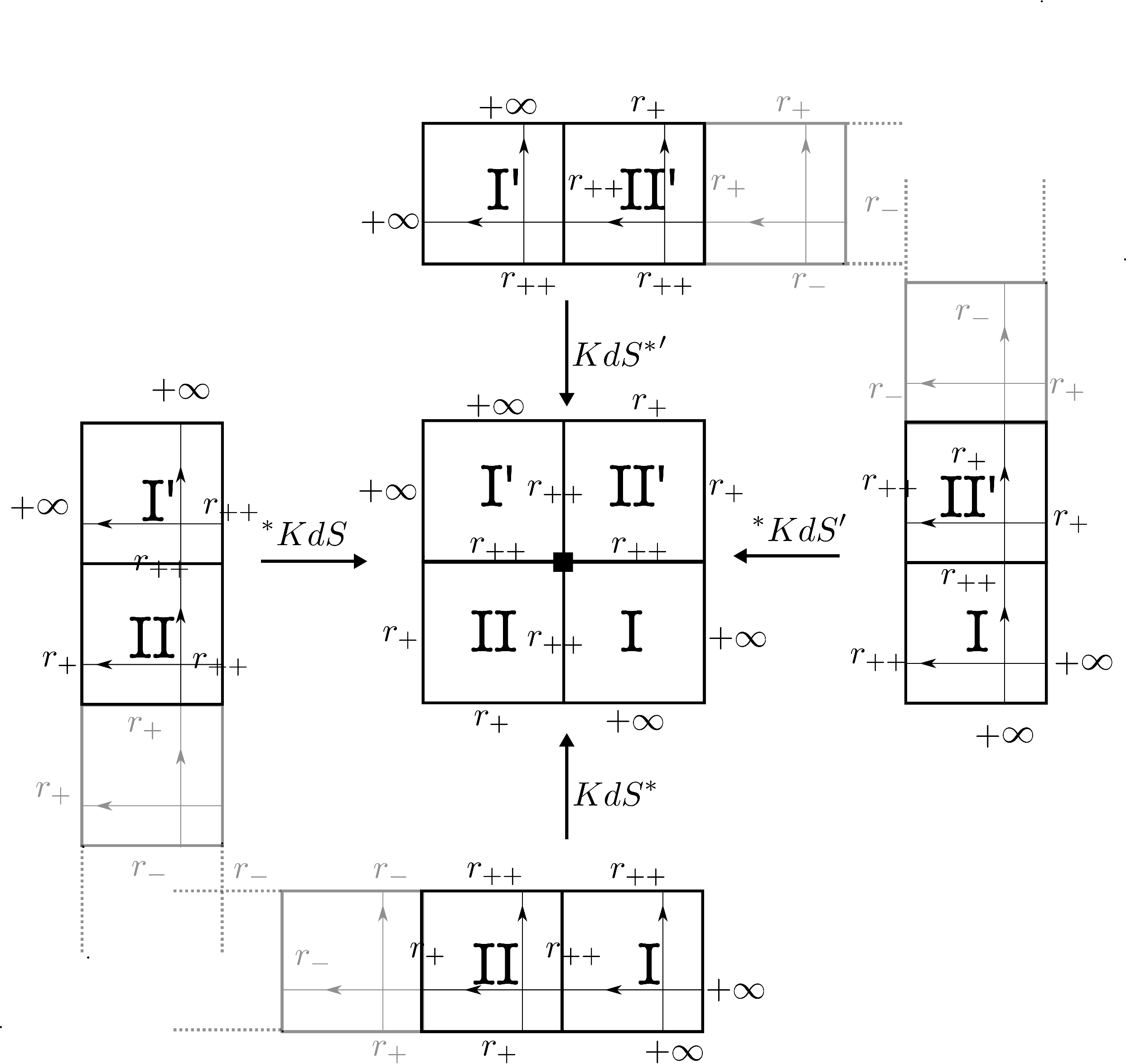}
\caption{Building $\mathscr{D}_0(r_{++})$ \label{figure:construct_dpp}}
\end{figure}
The details are as follows:
\begin{enumerate}
\item Begin with the manifold $K_1$ consisting of the open set containing blocks I$^*$ and II$^*$ in $KdS^*$. The ``outgoing" principal null geodesics of block I$^*$ are future-incomplete. In order to extend them, glue the open set of $\kstar KdS'$ containing blocks $\kstar$II and $\kstar$I onto $K_1$ using the time-orientation preserving isometry of section~\ref{section:starkds} to identify the blocks I$^*$ and $\kstar$I. It is necessary to use $\kstar KdS'$ as opposed to $\kstar KdS$ to ensure that the isometry preserves time-orientation. 
It may surprise the reader that, according to our terminology, we are extending an outgoing principal null geodesic using an ingoing principal null geodesic. This is not really the case, as inspection of figure~\ref{img:kstar} reveals that the ``outgoing" principle null geodesic of block I, is actually a badly named ``ingoing" principle null geodesic, since $dr^*(n+)\leq 0$ on block I.  

We verify briefly on this example that the condition of lemma~\ref{lemme:separation2} is satisfied:

Here the coordinate expression of $\phi: \text{I}^* \longrightarrow \kstar\text{I}$ is $$\phi(t^*,r^*,\theta^*,\phi^*) = (t^* - 2T(r^*),r^*,\theta^*; \phi^*-2A(r^*))$$
 Suppose that $(x_n)_{n\in\mathbb{N}}=(t^*_n,r^*_n,\theta^*_n, \phi^*_n)$ is a sequence of points in $U=I^*$ converging to a point on the horizon $r^*=r_{++}$, in particular the sequence $(t^*_n)_{n\in\mathbb{N}}$ has a finite limit, but $|T(r)| \underset{r\to r_{++}}{\longrightarrow} \infty$ so $(\phi(x_n))_{n \in \mathbb{N}}$ cannot converge.

\item Call $K_2$ the manifold obtained after step 1. We extend the outgoing principal null geodesics of block II in the same way, except that we use ${}^*KdS$, since on block II time-orientation is preserved by the isometry of \ref{section:starkds}.

\item Complete the manifold $K_3$ resulting from steps 1 and 2 by gluing the open set of ${KdS^*}'$  containing blocks I' and II' onto $K_3$ identifying, using the isometries of \ref{section:starkds}, I' \emph{and} II' with those contained in $K_3$.  
\end{enumerate}

\subsubsection{Crossing spheres}
\label{section:crossing_sphere}
Our ambition is to construct a spacetime in which all principal null geodesics are complete (except those that run into the singularity). Until now, we have payed very little attention to those which are trapped in the horizons. To fix notations, consider $KdS^*$, but this discussion also holds with very minor modifications in ${}^*KdS$. Recall from section~\ref{kds_star} that outgoing principal null geodesics are defined as geodesic reparametrisations of the integral curves of $n_+= \frac{\Delta_r}{2\Xi} \partial_{r^*} + V$. For any point $p$ on a horizon $\mathscr{H}$, $n_+(p)=V(p) \in T_p\mathscr{H}$.

\begin{lemme}
Let $i\in\{--,-,+,++\}$, then for any $p\in \mathscr{H}_i$: 

$$\left.(\nabla_V V)\right|_{p} = \frac{1}{\Xi}\left(r_i - M -l^2r_i(2r_i^2+a^2)\right)V$$
\end{lemme}

\begin{lemme}
Call $k_i=\frac{r_i - M -l^2r_i(2r_i^2+a^2)}{\Xi}, i \in \{--,-,+,++\}$ then: 
\begin{align}
&k_{++}=-\frac{l^2}{2\Xi}(r_{++}-r_{--})(r_{++}-r_{+})(r_{++}-r_{-})<0 \\
&k_{+}=\frac{l^2}{2\Xi}(r_+-r_{--})(r_{++}-r_+)(r_+-r_-)>0\\
&k_{-}=-\frac{l^2}{2\Xi}(r_{-}-r_{--})(r_{++}-r_{-})(r_{+}-r_{-})<0\\
&k_{--}=\frac{l^2}{2\Xi}(r_{++}-r_{--})(r_+-r_{--})(r_{-}-r_{--})>0
\end{align}
\end{lemme}
\begin{proof}
Follows immediately from the relation: $r_i-M-l^2r_i(2r_i^2+a^2)= \left.\frac{1}{2} \frac{\partial}{\partial r}\Delta_r\right|_{r=r_i}$ after factorisation of $\Delta_r$: $\displaystyle \Delta_r= -l^2\prod_i (r-r_i)$
\end{proof}

\begin{corollaire}
\label{cor:multiple_horizons}
Let $i \in \{--,-,+,++\}$, then, if $r_i$ is a root with multiplicity $>1$ of $\Delta_r$, then for any $p\in \mathscr{H}_i$:
$$\left.(\nabla_V V) \right|_p =0$$
\end{corollaire}

\begin{prop}
\begin{enumerate}
\item On horizons arising from a root of multiplicity $>1$ of $\Delta_r$, the integral curves of $V$ are complete. 
\item On the other horizons the integral curves of $V$ are not complete.
\end{enumerate}
\end{prop}

\begin{proof}
For the first point, according to corollary~\ref{cor:multiple_horizons} the integral curves of $n_+$ are already geodesically parametrised. Furthermore, since $V$ is a constant linear combination of the coordinate fields $\partial_{t^*}, \partial_{\phi^*}$, its integral curves are complete (i.e. they can be extended so that the interval of definition is $\mathbb{R}$).

Assume now that $r_i$ is a simple root of $\Delta_r$, then according to the above: $k_i \neq 0$, and the integral curves of $n_+$ are not geodesically parametrised.

A generic integral curve of $n_+$ on $\mathscr{H}_i$ is given in $KdS^*$ coordinates by:

\[\gamma(s)= ((r_i^2+a^2)s + t^*_0, r_i, \theta_0, as +\phi^*_0), s\in \mathbb{R}\]
Since $\partial_{\phi^*}$ and $\partial_{t^*}$ are global Killing fields on $KdS^*$, it suffices to consider the case where $t^*_0=\phi^*_0=0$. When geodesically parametrised and the affine parameter chosen so that $\tilde{\gamma} = \gamma \circ s(\lambda)$ is future-oriented, we have:
\begin{equation}\label{eq:integral_curve_v1} \tilde{\gamma}(\lambda)=\left((r_i^2+a^2) k^{-1}_i \ln (k_i\lambda),r_i,\theta_0, a k^{-1}_i \ln(k_i\lambda)\right),k_i\lambda>0 \end{equation} 
Which cannot be extended though $\lambda \to 0$.
\end{proof}

\begin{rem}

\begin{itemize}
\item On ${KdS^*}'$ where orientation is reversed, the future-oriented geodesic parametrisation of the integral curves is:

\begin{equation}\label{eq:integral_curve_v2} \tilde{\gamma}(\lambda)=\left((r_i^2+a^2) k^{-1}_i \ln (-k_i\lambda),r_i,\theta_0, a k^{-1}_i \ln(-k_i\lambda)\right),k_i\lambda<0, \end{equation}
\item The formulae for ${}^*KdS$ et ${}^*KdS'$ are obtained by the substitution : \[t^* \rightarrow {}^*t, \phi^* \rightarrow {} ^*\phi.\]
\end{itemize}
\end{rem}

Sending $\lambda \to 0$ in formulae~\eqref{eq:integral_curve_v1},\eqref{eq:integral_curve_v2}, it would seem that $\tilde{\gamma}(\lambda)$ approaches a point that would be located at the center of each of the diagrams of figure~\ref{figure:kruskal_slow}. We now seek to construct an analytic extension $\mathscr{D}(r_i)$ of each $\mathscr{D}_0(r_i)$ that contains such a limit point, this will be achieved by building a new system of coordinates.

\begin{definition}
\begin{align} A(r)&= \frac{a}{2\kappa_{--}}\ln|r-r_{--}| - \frac{a}{2\kappa_{-}}\ln |r-r_-| + \frac{a}{2\kappa_{+}}\ln |r-r_+| -\frac{a}{2\kappa_{++}}\ln |r-r_{++}| \\ T(r)& = \frac{r^2_{--} + a^2}{2\kappa_{--}} \ln |r-r_{--}| - \frac{r^2_-+a^2}{2\kappa_-}\ln |r-r_-| + \frac{r_+^2+a^2}{2\kappa_+}\ln |r-r_+| \\\nonumber&- \frac{r_{++}^2+a^2}{2\kappa_{++}}\ln|r-r_{++}| \\
\kappa_{i}&=\textrm{sgn}(k_i)k_{i} ,\quad i\in \{ --, - ,+ ,++\} 
\end{align}
\end{definition}
\begin{rem}
The quantity we denote $\kappa_{i}$ is normalised differently from the analogous one in~\cite{ONeill:2014aa}.
\end{rem}
The proofs of the following technical lemmata are left to the reader:
\begin{lemme}
For each $\displaystyle i\in \{ --,-,+,++\}, A(r)- \frac{a}{r^2_i+a^2} T(r)$ is analytic at $r_i$.
\end{lemme}

\begin{lemme}
Let $i\in \{--,-,+,++\}$:  
On any Boyer-Lindquist block (minus points on the axis $\mathcal{A}$), the functions $({}^*t,t^*,\theta,\phi^{i})$, where $\displaystyle \phi^{i}=\frac{1}{2}\left({}^*\phi + \phi^* - \frac{a}{r_{i}^2+a^2}(\kstar t + t^*) \right)$ form a coordinate chart.
\end{lemme}

We specialise now to $\mathscr{D}(r_{++})$:

\begin{definition}
Define maps $U^{++}, V^{++}$ on $\mathscr{D}(r_{++})$  by:
\begin{align*}
\text{On I'}  : \left\{\begin{array}{c} U^{++} = -\exp\left( \frac{\kappa_{++}\kstar t}{r_{++}^2+a^2} \right) \\ V^{++} = \exp\left(-\frac{\kappa_{++}t^*}{r^2_{++}+a^2}\right) \end{array} \right. && \text{On II}  : \left\{\begin{array}{c} U^{++} = -\exp\left( \frac{\kappa_{++}\kstar t}{r_{++}^2+a^2} \right) \\ V^{++} = -\exp\left(-\frac{\kappa_{++}t^*}{r^2_{++}+a^2}\right) \end{array} \right. \\
\text{On II'}  : \left\{\begin{array}{c} U^{++} = \exp\left( \frac{\kappa_{++}\kstar t}{r_{++}^2+a^2} \right) \\ V^{++} = \exp\left(-\frac{\kappa_{++}t^*}{r^2_{++}+a^2}\right) \end{array} \right. && \text{On I}  : \left\{\begin{array}{c} U^{++} = \exp\left( \frac{\kappa_{++}\kstar t}{r_{++}^2+a^2} \right) \\ V^{++} = -\exp\left(-\frac{\kappa_{++}t^*}{r^2_{++}+a^2}\right) \end{array} \right. 
\end{align*}
Recall that on I,I' $r>r_{++}$ and on II,II' $r_{+}<r<r_{++}$. 
\end{definition}
\begin{lemme}
\hfill
\begin{itemize}
\item $U^{++},V^{++}, \theta $ and $\phi^{++}$ have analytic extensions to all of $\mathscr{D}_0(r_{++})\setminus \{\text{axis points}\}$ (that we will denote by the same symbols). Furthermore $\eta^{++}=(U^{++}, V^{++}, \theta, \phi^{++})$ is a coordinate system on $\mathscr{D}_0(r_{++})\setminus \{\text{axis points}\}$
\item $\eta^{++}$ has an analytic extension to a diffeomorphism of $\mathscr{D}_0(r_{++})$ onto $\mathbb{R}^2\setminus \{(0,0)\} \times S^2$
\item $r$ has an analytic extension to all of $\mathbb{R}_{U^{++}}\times \mathbb{R}_{V^{++}} \times S^2$
\item $ r\mapsto G^{++}(r)= \frac{r-r_{++}}{U_{++}V_{++}}$ is an analytic function of $r\not\in\{r_-,r_+,r_{--} \}$ that never vanishes.
\end{itemize}
\end{lemme}

\begin{prop}
In the coordinates $\eta^{++}$ of $\mathscr{D}_0(r_{++}) \setminus \{\text{axis points}\}$, the line element can be expressed as:

\begin{align}\label{eq:dr++_metric}\dd s^2 = &\frac{\Delta_rG^{++}(r)^2}{r-r_{++}}\frac{r_{++}^2+a^2}{4\kappa_{++}^2\Xi^2\rho^2}\frac{r+r_{++}}{r^2+a^2}\left( \frac{\rho^2}{r^2+a^2} + \frac{\rho_{++}^2}{r_{++}^2+a^2} \right)a^2\sin^2\theta \\\nonumber& \hspace{2.8in}\times \left({V^{++}}^2\dd{U^{++}}^2 + {U^{++}}^2\dd{V^{++}}^2 \right)\\\nonumber &+g_{\theta\theta}\dd\theta^2+g_{\phi\phi}\dd\phi^2 \\\nonumber&+\frac{\Delta_rG^{++}(r)}{r-r_{++}}\frac{(r_{++}^2+a^2)^2}{2\kappa_{++}^2\rho^2\Xi^2}\left( \frac{\rho^4}{(r^2+a^2)^2} + \frac{\rho^4_{++}}{(r_{++}^2+a^2)^2} \right) \dd U^{++} \dd V^{++}\\\nonumber&+ \frac{a \sin^2\theta G^{++}(r)}{\rho^2\Xi^2 \kappa_{++}} \left( {\Delta_\theta (r+r_{++})(r^2+a^2)}+ \frac{\Delta_r\rho^2_{++}}{r-r_{++}}\right)\\\nonumber& \hspace{2.6in}\times \dd\phi^{++} \left( V^{++}\dd U^{++} - U^{++}\dd V^{++} \right)\\\nonumber&+\frac{\Delta_\theta a^2\sin^2\theta G^{++}(r)^2 (r+r_{++})^2}{4\kappa_{++}^2\rho^2\Xi^2}\left({V^{++}}\dd{U^{++}} -{U^{++}}\dd{V^{++}} \right)^2   \end{align}
where $\rho_{++}^2= r^2_{++} +a^2\cos^2\theta$. 
\end{prop}

\begin{proof}
We break down the calculation of expression~\eqref{eq:dr++_metric}, following similar steps as in~\cite{ONeill:2014aa} so that the reader can compare with the expressions in Kerr spacetime. One must keep in mind that our definition of $\kappa_{i}$ differs slightly from that of the analogous quantity in~\cite{ONeill:2014aa}. Our aim, more than to find the most compact expression, is to show that all apparent singularities cancel out. We shall express the $\eta^{++}$ in terms of the Boyer-Lindquist chart and change coordinates in expression~\eqref{eq:metricBL_compact}.

First note that:
\begin{equation}
\begin{gathered}
\phi^{++}= \phi - \frac{a}{r_{++}^2+a^2}t, \quad t = \frac{r_{++}^2+a^2}{2\kappa_{++}}\ln\left( \frac{|U^{++}|}{|V^{++}|} \right),\\ T(r)= -\frac{r_{++}^2+a^2}{2\kappa_{++}}\ln\left( |U^{++}||V^{++}| \right).
\end{gathered}
\end{equation}
From which it follows that:
\begin{equation}
\phi= \phi^{++} + \frac{a}{2\kappa_{++}}\ln \left( \frac{|U^{++}|}{|V^{++}|} \right).
\end{equation}
Therefore:
\begin{equation}
\begin{cases}
\dd t =\frac{r_{++}^2+a^2}{2\kappa_{++}}\frac{G^{++}(r)}{r-r_{++}}\left(V^{++}\dd U^{++} - U^{++}\dd V^{++} \right),\\
\dd \phi = \dd\phi^{++} + \frac{a}{2\kappa_{++}}\frac{G^{++}(r)}{r-r_{++}}\left( V^{++} \dd U^{++} - U^{++}\dd V^{++} \right),\\
\dd r = -\frac{\Delta_r}{2\Xi\kappa_{++}}\frac{r_{++}^2+a^2}{r^2+a^2}\frac{G^{++}(r)}{r-r_{++}}\left( V^{++}\dd U^{++} + U^{++}\dd V^{++} \right).
\end{cases}
\end{equation}
The remainder of the computation consists in injecting these expressions into~\eqref{eq:metricBL_compact}. We first evaluate the terms in brackets:
\[ \begin{cases}\dd t - a\sin^2\dd \phi= \frac{\rho_{++}^2}{2\kappa_{++}}\frac{G^{++}(r)}{r-r_{++}}\left(V^{++}\dd U^{++} - U^{++}\dd V^{++}\right) - a\sin^2\theta \dd\phi^{++}, \\ (r^2+a^2)\dd\phi -a\dd t = (r^2+a^2)\dd \phi^{++} + \frac{aG^{++}(r)}{2\kappa_{++}}(r+r_{++})(V^{++}\dd U^{++} - U^{++} \dd V^{++}). \end{cases}\]
Whilst the second term squared and multiplied by $\frac{\Delta_\theta \sin^2\theta}{\rho^2\Xi^2}$ is regular away from the ring singularity, the first term squared and multiplied by $\frac{-\Delta_r}{\rho^2\Xi^2}$ leads to terms with an apparent singularity at $r=r_{++}$ as $r_{++}$ is a simple root of $\Delta_r$. These terms are compensated by:
\[ \frac{\rho^2}{\Delta_r}\dd r^2 = \frac{(r_{++}^2+a^2)^2}{4\Xi^2\kappa^2_{++}}\frac{G^{++}(r)^2}{(r^2+a^2)^2}\frac{\rho^2 \Delta_r}{(r-r_{++})^2}\left(V^{++}\dd U^{++} + U^{++}\dd V^{++} \right)^2. \]
Isolating the appropriate part of $-\frac{\Delta_r}{\Xi^2\rho^2}[\dd t - a\sin^2\theta \dd \phi]^2$ leads us to consider the expression:
\[A_1= \frac{\rho^2}{\Delta_r}\dd r^2 - \frac{\Delta_r}{\Xi^2\rho^2}\frac{\rho_{++}^4G^{++}(r)^2}{4\kappa_{++}^2(r-r_{++})^2}\left( V^{++}\dd U^{++} - U^{++}\dd V^{++}\right)^2. \]
Using that:
\[ \frac{\rho^2}{r^2+a^2} - \frac{\rho_{++}^2}{r_{++}^2+a^2}= \frac{a^2\sin^2\theta (r-r_{++})(r+r_{++})}{(r^2+a^2)(r_{++}^2+a^2)}, \]
one finds that $A_1$ can be written:
\begin{equation}
A_1=C_1\left((U^{++})^2{\dd V^{++}}^2 + (V^{++})^2{\dd U^{++}}^2\right)+C_2 \dd U^{++}\dd V^{++},
\end{equation}
where:
\begin{gather}
C_1=\frac{\Delta_rG^{++}(r)^2(r_{++}^2+a^2)(r+r_{++})}{(r-r_{++})4\Xi^2\kappa_{++}^2\rho^2(r^2+a^2)}\left(\frac{\rho^2}{r^2+a^2}+\frac{\rho_{++}^2}{r_{++}^2+a^2} \right)a^2\sin^2\theta,
\\ C_2=\frac{\Delta_r G^{++}(r)(r_{++}^2+a^2)^2}{(r-r_{++})2\Xi^2\kappa_{++}^2\rho^2}\left(\frac{\rho^4}{(r^2+a^2)^2}+\frac{\rho_{++}^4}{(r_{++}^2+a^2)^2} \right).
\end{gather}
Expanding $\frac{\Delta_\theta\sin^2\theta}{\rho^2\Xi^2}[(r^2+a^2)\dd \phi - a \dd t]^2$ and combining with the previous computation, we arrive at Equation~\eqref{eq:dr++_metric}.
\end{proof}
%
Expression~\eqref{eq:dr++_metric} extends analytically to all of $(\mathbb{R}_{U^{++}}\times \mathbb{R}_{V^{++}} )\times S^2$ and it is straightforward to verify that it is non-degenerate at points of $\{(0,0)\}\times S^2$.
This concludes the construction of $\mathscr{D}(r_{++})$ which is defined as $(\mathbb{R}_{U^{++}}\times \mathbb{R}_{V^{++}})\times S^2$ equipped with the metric~\eqref{eq:dr++_metric}. Similar expressions for the metric can be obtained on the other Kruskal domains. We can now check that these extra points really do enable the extension of incomplete principal null geodesics contained in the horizons by welding together those from the different Boyer-Lindquist blocks. 
Recall from equation~\eqref{eq:integral_curve_v2} the geodesic parametrisation of a generic integral curve, expressed in $KdS^*$ coordinates, contained in the horizon $\mathscr{H}_i$ and coming from ${KdS^*}'$ (see figure~\ref{figure:construct_dpp}): 

$$\tilde{\gamma}(\lambda)=\left((r_{++}^2+a^2) k^{-1}_{++} \ln (-k_{++}\lambda),r_i,\theta_0, a k^{-1}_{++} \ln(-k_{++}\lambda)\right),\lambda>0,$$

This curve is past-incomplete and its expression in Kruskal coordinates is:

\begin{equation}
\left\{\begin{array}{l} U^{++}=0\\
V^{++}=  -k_{++} \lambda\\
\theta=\theta_0\\
\phi^{++}= -\displaystyle \lim_{r\to r_{++}} A(r)-\frac{a}{r_{++}^2+a^2}T(r)
\end{array}\right. \lambda \in \mathbb{R}^*_+
\end{equation}

From these expressions we see that when $\lambda \to 0$, $\gamma$ approaches a point on the crossing-sphere ($U^{++}=V^{++}=0$)

If we consider now a similar curve in the horizon coming from ${KdS^*}$, then its geodesic parametrisation in $KdS^*$ coordinates is, from~\eqref{eq:integral_curve_v1}:

$$\tilde{\gamma}(\lambda)=\left((r_{++}^2+a^2) k^{-1}_{++} \ln (k_{++}\lambda),r_{++},\theta_0, a k^{-1}_{++} \ln(k_{++}\lambda)\right),\lambda <0  $$

This curve is future incomplete; converting to Kruskal coordinates:

\begin{equation}
\left\{\begin{array}{l} U^{++}=0\\
V^{++}=  -k_{++} \lambda\\
\theta=\theta_0\\
\phi^{++}= -\displaystyle \lim_{r\to r_{++}} A(r)-\frac{a}{r_{++}^2+a^2}T(r)
\end{array}\right. \lambda \in \mathbb{R}^*_-
\end{equation}

The curves clearly analytically extend one another to form a complete geodesic. Through this example, we see that the role of the crossing-sphere ($U^{++}=V^{++}=0$) really is to join together the two ``vertical'' horizons in figure~\ref{figure:construct_dpp} to form a single null hypersurface of equation $U^{++}=0$. The results are similar when considering the principal null geodesics in the ``horizontal'' horizons of figure~\ref{figure:construct_dpp}.
\subsubsection{Building maximal slow Kerr-de Sitter $KdS_s$}
We will now describe how to combine the Kruskal domains of section~\ref{slow_kruskal} to build the maximal slow Kerr-de Sitter spacetime $KdS_s$; the gluing pattern is illustrated in figure~\ref{fig:kdsslowpattern}.
\begin{figure}[h]
\centering
\includegraphics[scale=.35]{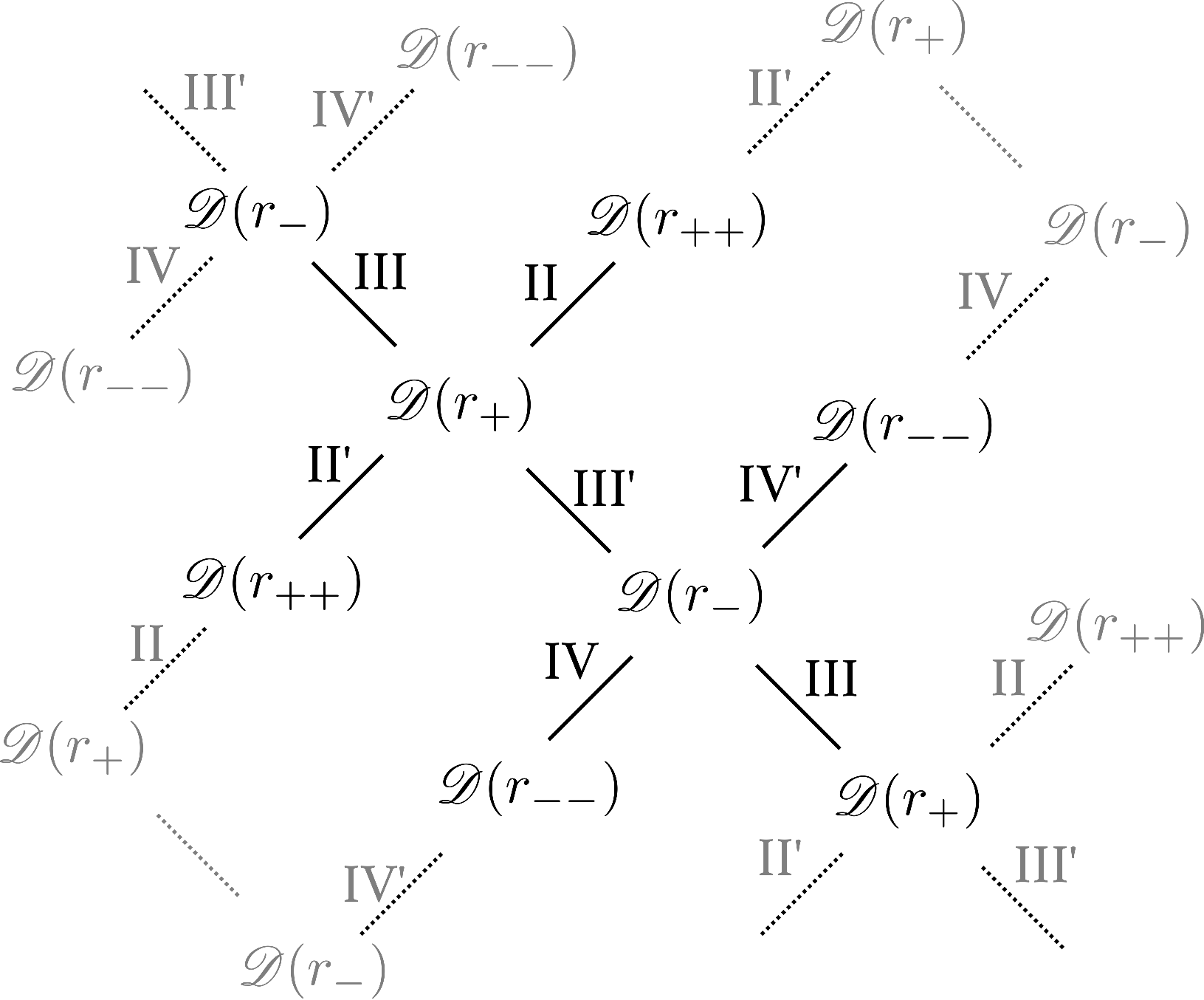}
\caption{Gluing pattern to construct $KdS_s$; the roman numeral labels indicate which Boyer-Lindquist block is used for the gluing \label{fig:kdsslowpattern}}
\end{figure}

To realise the gluing, begin with the two manifolds $K_1,K_2$ defined by:

\begin{itemize}
\item $K_1$ is the manifold obtained by considering two sequences $(D^+_i)_{i \in \mathbb{Z}}, (D^-_j)_{j\in \mathbb{Z}}$ of isometric copies of $\mathscr{D}(r_+)$ and $\mathscr{D}(r_-)$ respectively. Define: $\displaystyle X = \coprod_i D^+_i$, $\displaystyle Y = \coprod_j D^-_j$. We introduce some notations useful in the sequel: 
\begin{itemize} \item For each $k\in \mathbb{Z}$ denote by $i^k_+ : D^+_k\simeq \mathscr{D}(r_+) \rightarrow X$ and $i^k_-: D^-_k \simeq \mathscr{D}(r_-) \rightarrow Y$ the canonical injections. \item For any Boyer-Lindquist block $\mathcal{B}\subset \mathscr{D}(r_{\pm})$, $\mathcal{B}^{\pm}_i$ will denote the image of that block by the isometry $\mathscr{D}(r_{\pm})\simeq D^{\pm}_i $. \item $\bm{\mathcal{B}}^{\pm}_i = i_{\pm}^i(\mathcal{B}^{\pm}_i)$ \end{itemize} Define now\footnote{see appendix~\ref{app:gluing} }: $K_1= X \coprod_\phi Y$ where $\phi : \coprod_i III_i \cup III_i'  \rightarrow Y$ is constructed using the universal property of coproducts from the maps:

$$\phi_i:  III_i \cup III'_i \subset D^+_i \longrightarrow \bm{III_i}\cup \bm{III'_{i-1}} \subset Y$$ 

which, when restricted to $III_i$ (resp. $III'_i$) and expressed in Boyer-Lindquist coordinates, is simply the identity map.
\item $\displaystyle K_2= (\coprod_i D^{++}_i) \coprod (\coprod_j D^{--}_j)$ is the disjoint union of the sequences $(D^{++}_i)_{i \in \mathbb{Z}}, (D^{--}_j)_{j\in \mathbb{Z}}$ of isometric copies of $D^{++}_i \simeq \mathscr{D}(r_{++})$ and $D^{--}_j\simeq \mathscr{D}(r_{--})$. 
\end{itemize}

As illustrated in~\ref{fig:kdsslowpattern}, $KdS_s$ can be built from $K_1$ and $K_2$ by gluing infinitely many copies of these manifolds along blocks with the same label. More precisely, consider two sequences $(M_i)_{i\in \mathbb{Z}}$ and $(N_j)_{j \in \mathbb{Z}}$ of manifolds. This time, for each $i \in \mathbb{Z}$, $M_i$ (resp. $N_i$) is an isometric copy of $K_1$ (resp. $K_2$). Define $\tilde{X}=\coprod_i M_i, \tilde{Y}=\coprod_j M_j$ and denote by $I_i: M_i \rightarrow \tilde{X}$ and $J_i: N_i : \rightarrow \tilde{Y}$ the canonical injections. $KdS_s$ will then be $\tilde{X} \coprod_\psi \tilde{Y}$ for a well chosen isometry $\psi$.

$\psi$ can be specified in several stages from maps $(\psi^{\pm\, i}_k)_{(i, k) \in \mathbb{Z}^2}$:

$$\psi^{+\,i}_k : II_k \cup II'_k \subset D^+_k \longrightarrow \bm{II^{++}_{(i,k)}} \cup \bm {II'^{++}_{(i-1,k)}} \subset \tilde{Y}$$

$$\psi^{-,\,i}_k : IV'_k \cup IV'_k \subset D^-_k \longrightarrow \bm{{IV'}^{++}_{(i,k)}}\cup \bm{IV^{--}_{(i-1,k)}} \subset \tilde{Y}$$

Where, $\bm{II^{++}_{(i,k)}}= J_i \circ i^{++}_{(i,k)}(II)$ and $i^{++}_{(i,k)}$ is the canonical injection of $D^{++}_k$ into $N_i$; the other sets are defined similarly. Again, when restricted to a given Boyer-Lindquist block and expressed in Boyer-Lindquist coordinates, these are just the identity maps. Using a natural generalisation of point 3 of proposition~\ref{prop:gluing1} in appendix~\ref{app:gluing}, for every $i\in \mathbb{N}$ this specifies a map:

$$\psi^{i} : \bigcup_{k \in \mathbb{Z}} \bar{i}^{(i,k)}_{+}(II_k \cup II'_k) \cup \bar{i}^{(i,k)}_-(IV_k \cup IV'_k) \subset M_i \rightarrow \tilde{Y}$$

These maps, using the universal property of coproducts, define together an isometry from:

$$\psi : \coprod_{i \in \mathbb{Z}}\bigcup_{k \in \mathbb{Z}} \bar{i}^{(i,k)}_{+}(II_k \cup II'_k) \cup \bar{i}^{(i,k)}_-(IV_k \cup IV'_k) \subset M_i \rightarrow \tilde{Y}$$

\subsection{Maximal extreme and fast KdS spacetimes}
Straightforward adaptations of the techniques of the previous section enable us to construct the maximal extreme and fast KdS spacetimes. For the extreme spacetimes, as discussed in section~\ref{delta_r}, there are three cases: $r_+=r_-$, $r_{++} = r_+$ or $r_{++}=r_+=r_-$.
\subsubsection{$KdS_e^1: r_+=r_-$}
We begin with the case where the two black hole horizons coincide and in which the Boyer-Lindquist block III disappears. The Kruskal domains $\mathscr{D}(r_{--})$ and $\mathscr{D}(r_{++})$ are unchanged, but the domains $\mathscr{D}(r_+)$ and $\mathscr{D}(r_-)$ are to be replaced by the domains $I_1$ and $I_2$ given in figure~\ref{figure:kruskal_extreme1}. The form of these domains can be understood from the fact that the horizon $\mathscr{H}_+$ now arises from a double root and the principal null geodesics trapped in it are complete; in particular there are no crossing spheres on the double horizons.
\begin{figure}[h]
\centering
\begin{subfigure}{.4\textwidth}
\includegraphics[scale=.3]{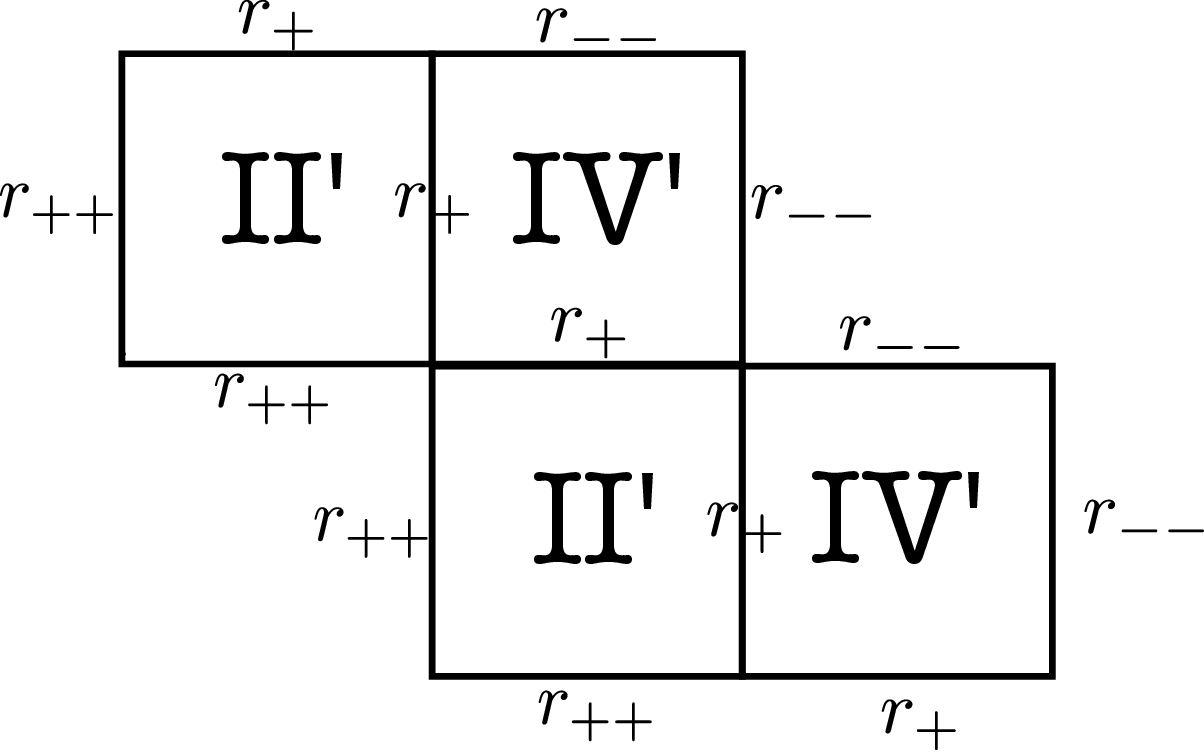}
\caption{$I_1$}
\end{subfigure}
\begin{subfigure}{0.4\textwidth}
\includegraphics[scale=.3]{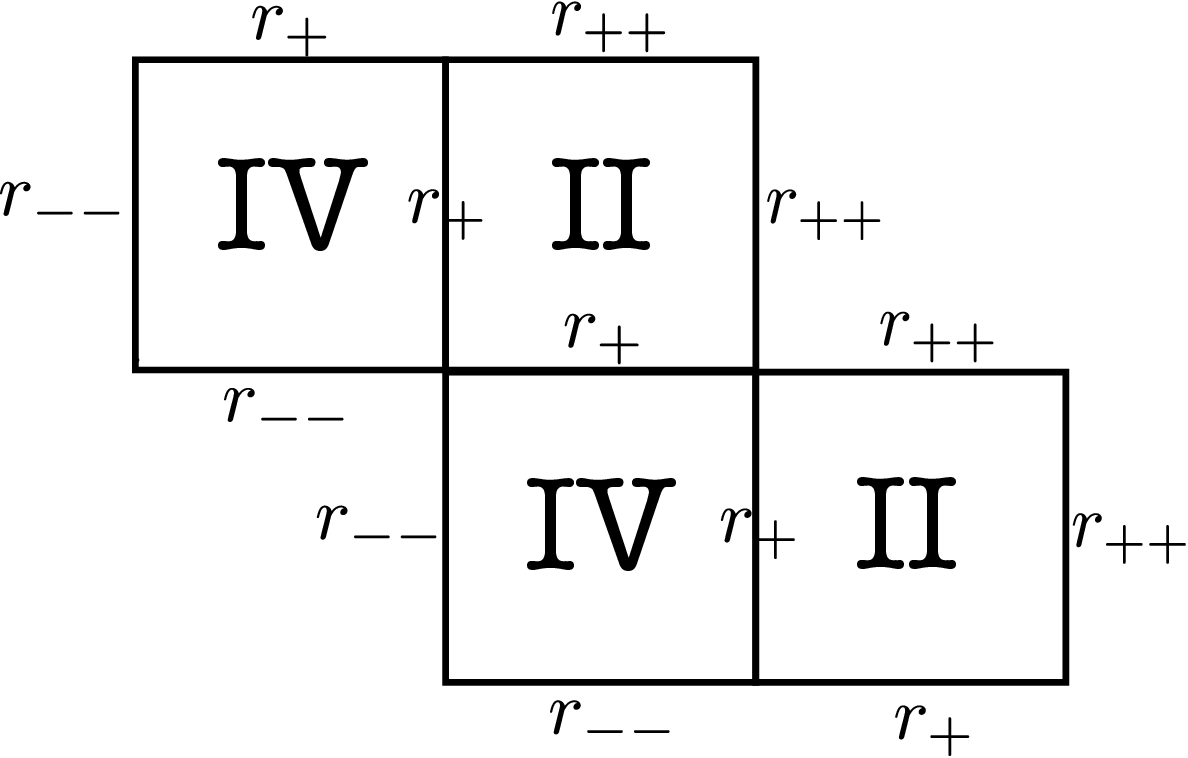}
\caption{$I_2$}
\end{subfigure}
\caption{Kruskal domains \label{figure:kruskal_extreme1}}
\end{figure}
The slightly simpler gluing pattern is illustrated in figure~\ref{figure:kdsextreme1pattern}. As before, the roman numeral labels indicate the blocks that are identified.
\begin{figure}[h]
\centering
\includegraphics[scale=.35]{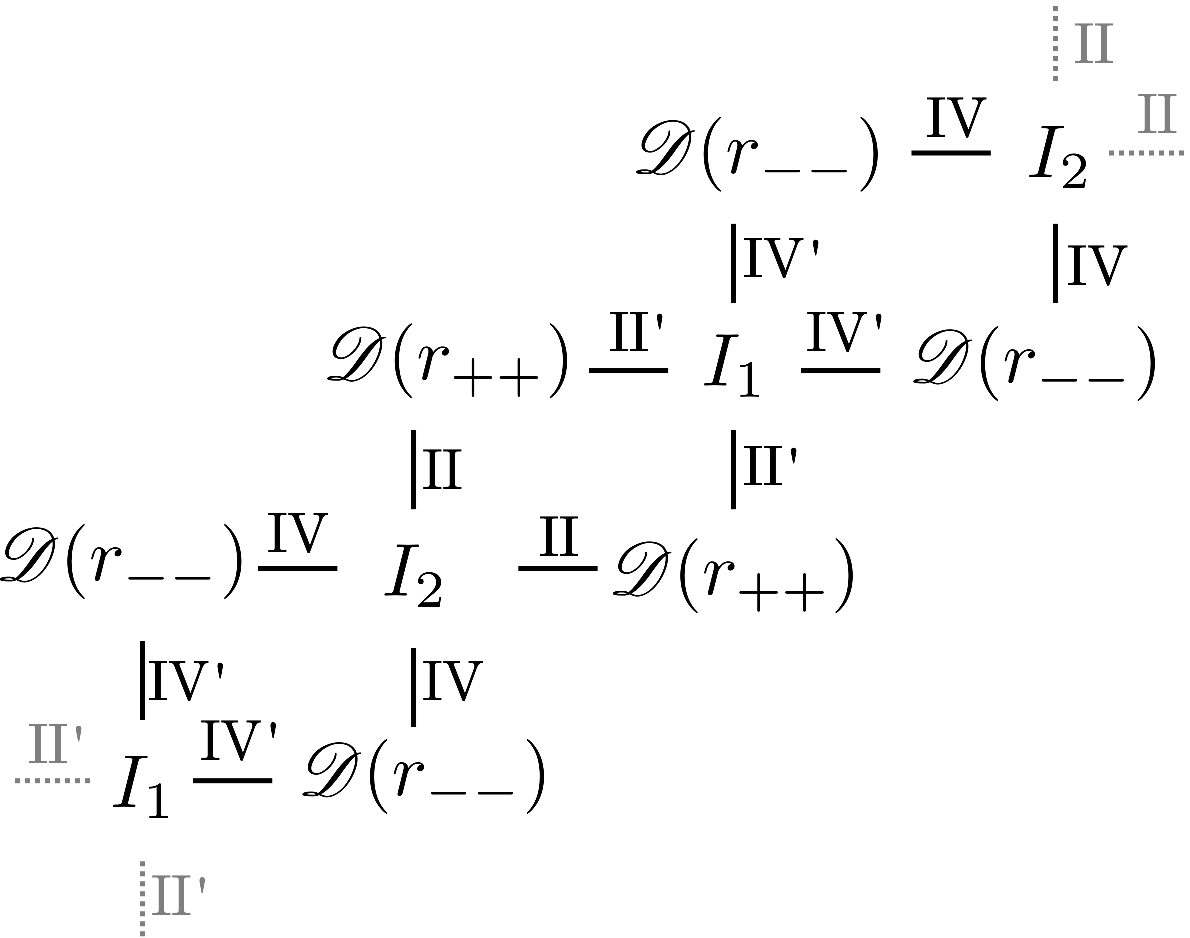}
\captionsetup{justification=centering}
\caption{Gluing pattern for $KdS_e$ \\ $r_+=r_-$ \label{figure:kdsextreme1pattern}}
\end{figure}

\subsubsection{$KdS_e^2: r_+=r_{++}$}

The second case is when the cosmological horizon $r_{++}$ coincides with the outer black hole horizon $r_{+}$. Here the Kruskal domains $\mathscr{D}(r_{--})$ and $\mathscr{D}(r_-)$ are unchanged and the remaining blocks are replaced by the domains illustrated in figure~\ref{figure:kruskal_extreme2}.
\begin{figure}[h]
\centering
\begin{subfigure}{.4\textwidth}
\includegraphics[scale=.3]{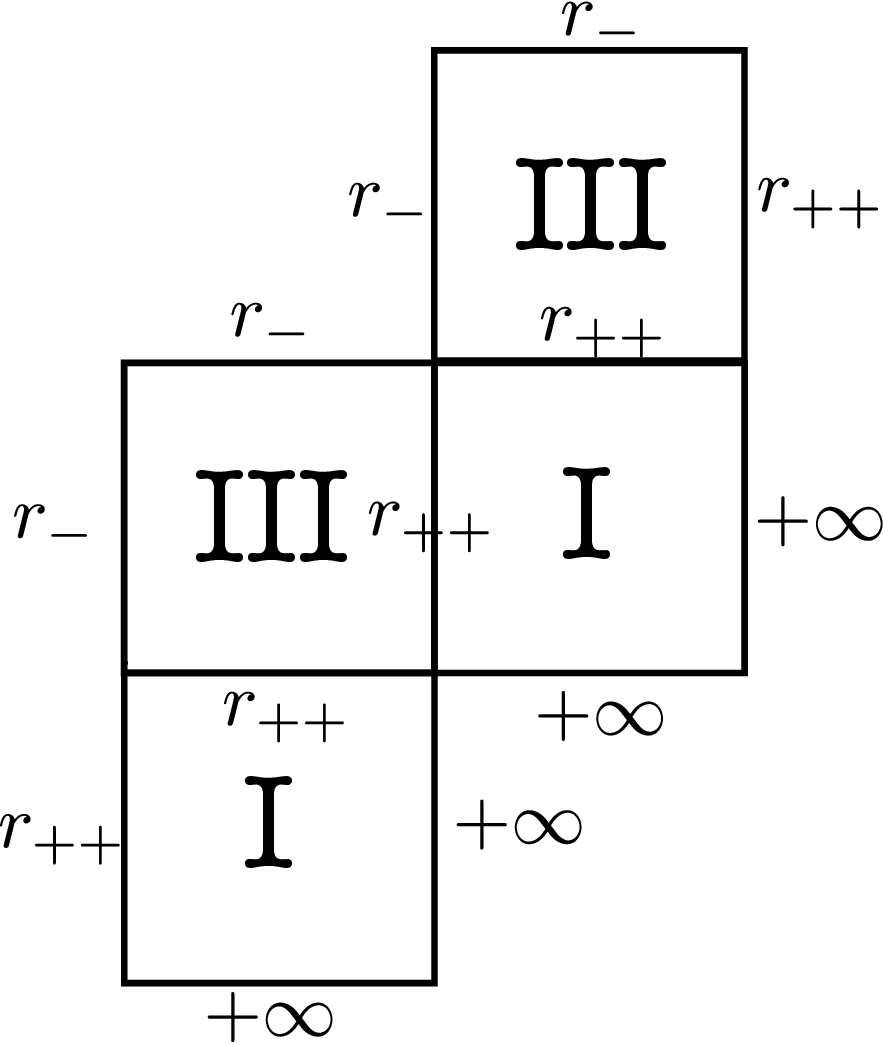}
\caption{$I_1$}
\end{subfigure}
\begin{subfigure}{0.4\textwidth}
\includegraphics[scale=.3]{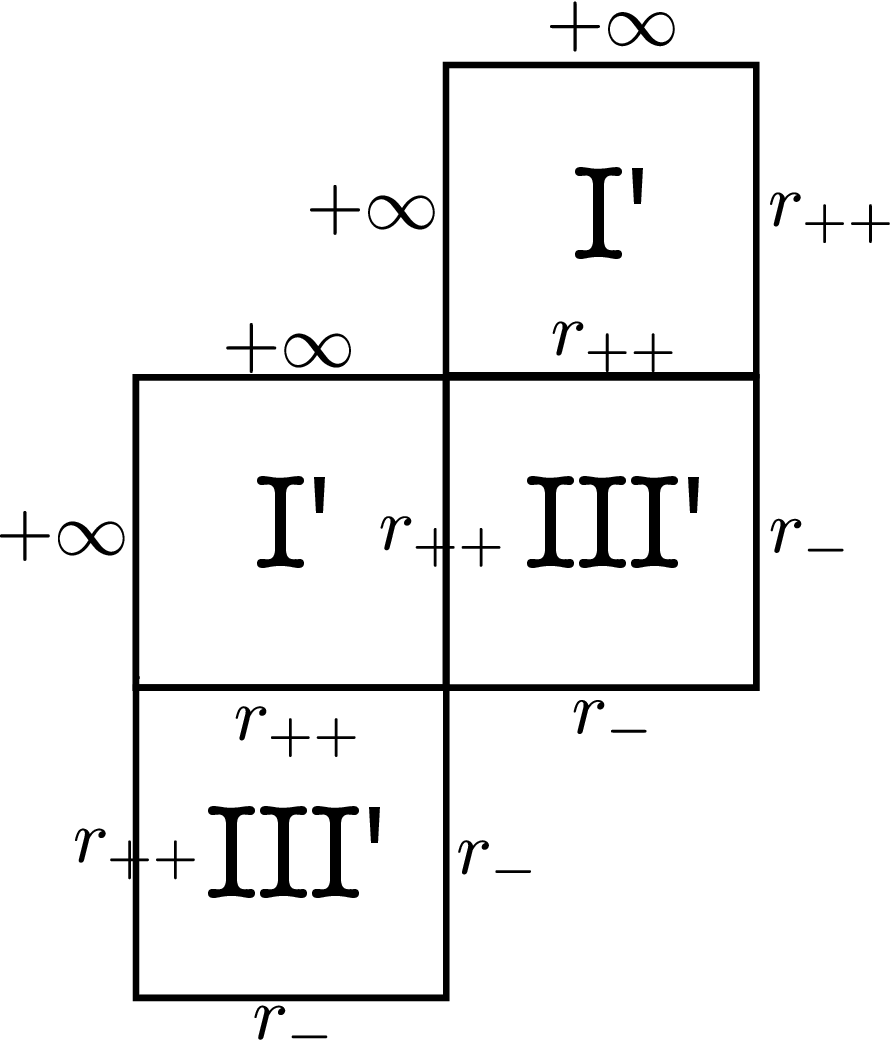}
\caption{$I_2$}
\end{subfigure}
\caption{Kruskal domains \label{figure:kruskal_extreme2}}
\end{figure}
The stranger gluing pattern is illustrated in figure~\ref{figure:kdsextreme2pattern}.
\begin{figure}[h]
\centering
\includegraphics[scale=.24]{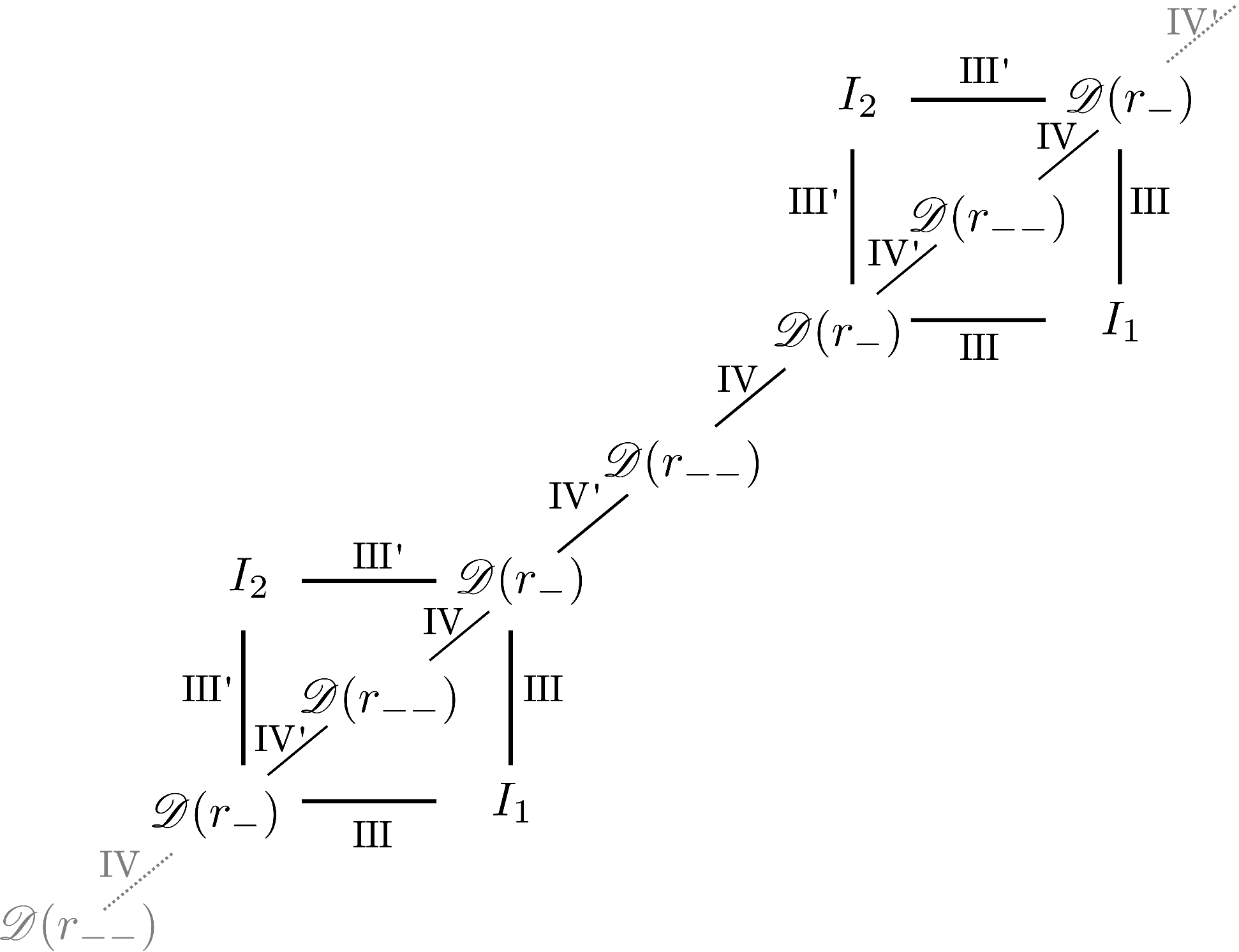}
\captionsetup{justification=centering}
\caption{Gluing pattern for $KdS_e^2$\\ $r_{++}=r_{+}$ \label{figure:kdsextreme2pattern}}
\end{figure}

\subsubsection{$ KdS_e^3: r_{++}=r_+=r_{-}=x$}
When $\Delta_r$ has a triple root $x$, we saw previously that all the horizons in the region $r>0$ coincide; Boyer-Lindquist blocks II and III consequently vanish. Contrary to the other cases, only two Kruskal domains are required to construct a maximal extension: the domain $\mathscr{D}(r_{--})$, as illustrated in~\ref{figure:kruskal_slow}, and the domain $\mathscr{D}_0(x)\equiv \mathscr{D}(r_{++})$ illustrated in figure~\ref{figure:racinetriple}.
\begin{figure}[h]
\centering
\includegraphics[scale=.3]{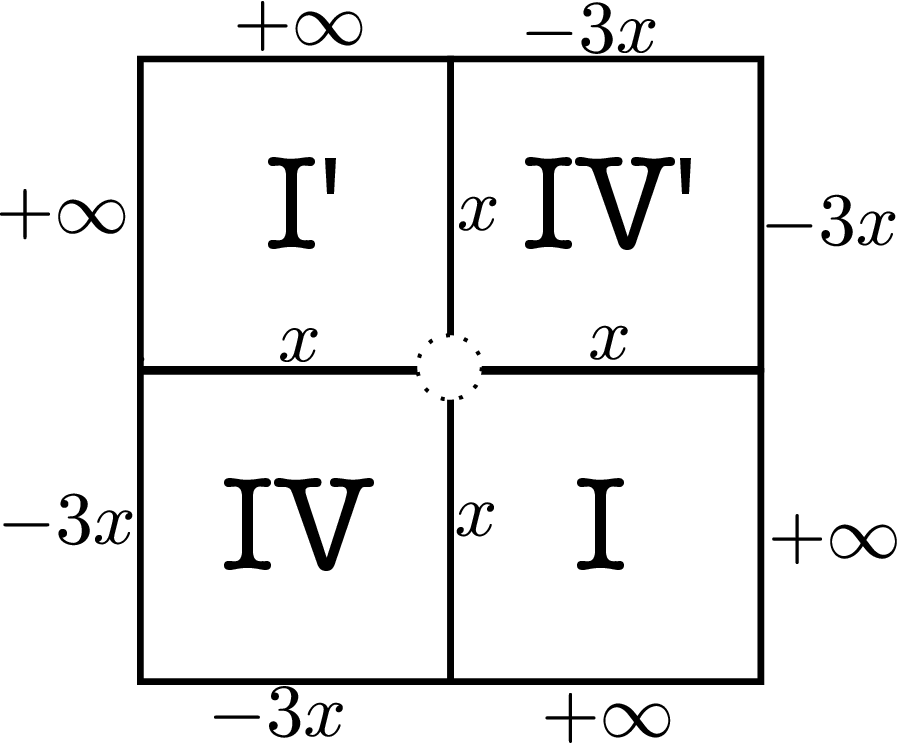}
\caption{ $\mathscr{D}_0(x) \equiv \mathscr{D}(r_{++})$ \label{figure:racinetriple} }
\end{figure}

Diagram~\ref{figure:racinetriple} has a striking ressemblance to that of $\mathscr{D}(r_{++})$ in figure~\ref{figure:kruskal_slow}, but is profoundly different due to the absence of the crossing sphere. Hence, whilst correctly depicting the assembly process leading to $\mathscr{D}_0(x)$, it is misleading for the interpretation of the geometry. In particular, like for the double horizons, Kruskal coordinates do not have analytic extensions to the whole domain. 

As expected, the gluing pattern for $KdS_e^3$, illustrated in figure~\ref{figure:gluing_kdse3}, is much simpler than in the other cases due to the fewer number of horizons and Boyer-Lindquist blocks.


\begin{figure}[h]
\centering
\includegraphics[scale=.35]{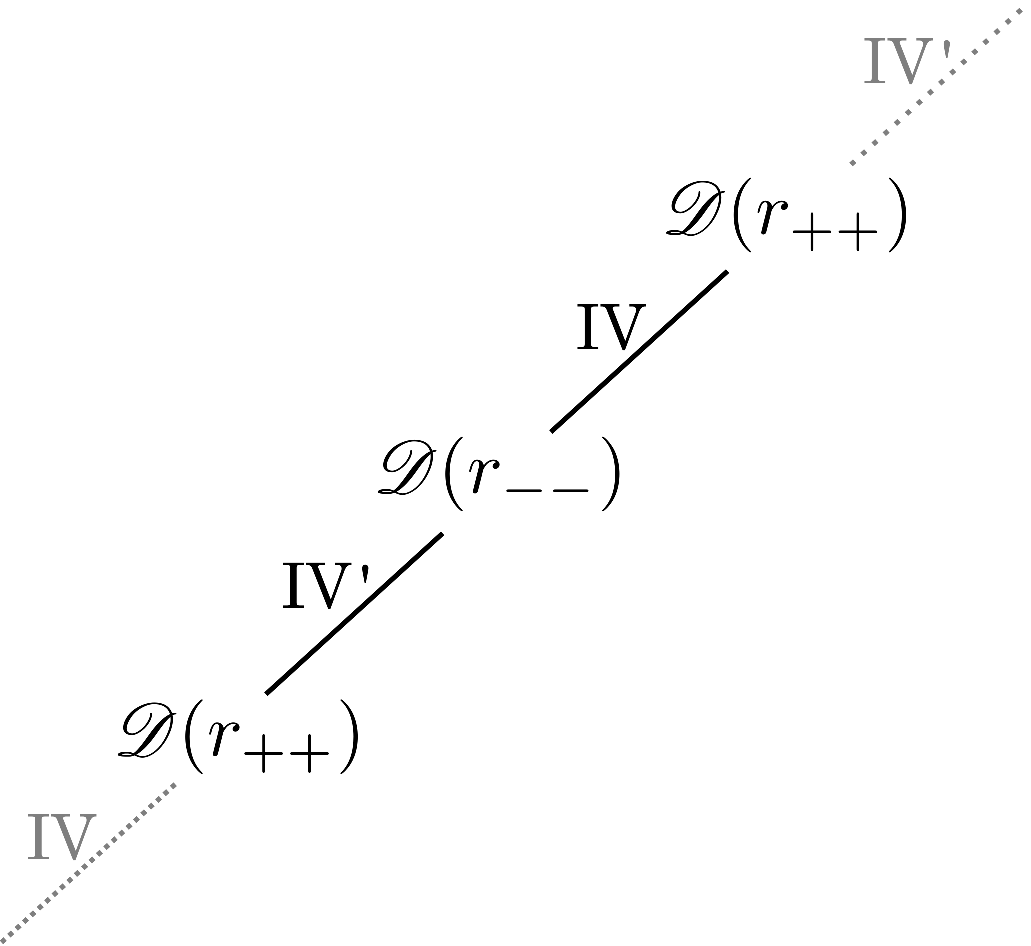}
\captionsetup{justification=centering}
\caption{\label{figure:gluing_kdse3}Gluing pattern for $KdS_e^3$ \\ $r_{++}=r_{+}=r_{-}=x \quad r_{--}=3x$}
\end{figure}

\subsubsection{Maximal Fast KdS spacetimes}
This final case, where $\Delta_r$ has only two simple real roots $r_{--}$ and $r_{++}$, is in all points analogous to slow Kerr-spacetime as presented in~\cite{ONeill:2014aa}; the main qualitative difference is that time orientation is reversed. There are only two Kruskal domains, $\mathscr{D}(r_{++})$ and $\mathscr{D}(r_{--})$ as illustrated in figure~\ref{figure:kruskal_slow}, with the exception that, due to the absence of blocks $II$ and $III$ , labels $II$ and $II'$ in figure~\ref{figure:kruskal_slow} should be replaced by $IV$ and $IV'$ respectively. The gluing pattern is identical to that in figure~\ref{figure:gluing_kdse3}.

\section{Conclusion}

The aim of this rather technical note was to give a detailed mathematical discussion regarding the construction of maximal analytical extensions to the Kerr-de Sitter solution to Einstein's equation with cosmological constant, as well as a review of the basic geometric properties of these spacetimes. The latter discussion can be found in~\ref{section:kds_metric}. To the best of the author's knowledge, in existing literature, the construction is only briefly commented upon and is not carried out explicitly as in section~\ref{maximal}.

Section~\ref{delta_r} is devoted to the study of the roots of the polynomial $\Delta_r$ in terms of the parameters $(a,l,M)$, and hence, the horizon structure of the blackhole. The referees brought to the attention of the author that similar discussions, although less mathematical, are present in earlier publications, namely~\cite{Stuchlik:2004aa} for Kerr-de Sitter and \cite{Stuchlik:2000aa} for the more general situation of Kerr-Newmann black holes on a background with non-zero cosmological constant.

\section*{Acknowledgements}

I would like to thank Jean-Philippe Nicolas for encouraging me to submit this note and the referees for their comments and advice.
\bibliographystyle{plain}
\bibliography{biblio.bib}


\newpage
\appendix
\section{Connection forms}
\label{app_connection_forms}
\begin{equation}
\label{eq:connection_forms}
\begin{split}
\omega^0_{\,\,\, 1} &= F \omega^0 - \frac{\varepsilon a r }{\rho^3}\sqrt{\Delta_\theta}\sin \theta \omega ^3 \\
\omega^0_{\,\,\, 2} &= -\frac{\sqrt{\Delta_\theta}a^2\sin\theta \cos \theta}{\rho^3} \omega^0 - \frac{\sqrt{\varepsilon \Delta_r} a \cos \theta}{\rho^3}\omega^3 \\
\omega^0_{\,\,\, 3} &= \frac{\sqrt{\varepsilon\Delta_r} a \cos \theta}{\rho^3}\omega^2 - \frac{\varepsilon a r \sqrt{\Delta_\theta} \sin \theta}{\rho^3}\omega^1 \\ 
\omega^1_{\,\,\, 2} &= -\frac{a^2\sin\theta\cos\theta \sqrt{\Delta_\theta}}{\rho^3} \omega^1 - \varepsilon r\frac{\sqrt{\varepsilon\Delta_r}}{\rho^3}\omega^2 \\
\omega^1_{\,\,\, 3} &= -\varepsilon a r \sin \theta \frac{\sqrt{\Delta_\theta}}{\rho^3} \omega^0 - \varepsilon r \frac{\sqrt{\varepsilon\Delta_r}}{\rho^3}\omega^3\\
\omega^2_{\,\,\, 3} &= -a \cos\theta \varepsilon \frac{\sqrt{\varepsilon \Delta_r}}{\rho^3} \omega^0 - \left( \textrm{cotan}\theta(r^2+a^2)\frac{\sqrt{\Delta_\theta}}{\rho^3} + \frac{G}{\rho} \right)\omega^3
\end{split}
\end{equation}

Where: $F=\frac{\partial}{\partial r} \left(\frac{\sqrt{\varepsilon \Delta_r}}{\rho} \right)$ and $G=\frac{\partial}{\partial \theta} \left(\sqrt{\Delta_\theta}\right)$

%


\section{Geodesic equations ``à la Cartan"}
\label{app:geodesic_equations}
Let $\gamma: I \longrightarrow KdS$, be a curve on one of the Boyer-Lindquist blocks of Kerr-de Sitter spacetime. Decomposing on the orthonormal frame one has at each point $t\in I$, $\dot{\gamma}(t)= \Gamma^i(t) E_i(\gamma(t)) \equiv \Gamma^i(t)E_i(t)$, so:

\begin{align*} \frac{D}{dt} \dot{\gamma} (t) =(\nabla_{\dot{\gamma}} \dot{\gamma})_{\gamma(t)}&= \dot{\Gamma^i}(t)E_i(t) + \Gamma^{i}(t)\Gamma^{j}(t)(\nabla_{E_i}E_j)_{\gamma(t)} \\
&= \dot{\Gamma^i}(t)E_i(t) + \Gamma^{k}(t)\Gamma^{j}(t)(\omega^{i}_{\,\, \,j})_{\gamma(t)}(E_k(t))E_i(t)  \end{align*}

If $\gamma$ is a geodesic, using~\eqref{eq:connection_forms} we find that the components satisfy the following system of differential equations:

\begin{align*} &\dot{\Gamma^0} + F\Gamma^0\Gamma^1 - 2\varepsilon a r \sin \theta \frac{\sqrt{\Delta_\theta}}{\rho^3}\Gamma^1\Gamma^3 -a^2\sin\theta\cos\theta\frac{\sqrt{\Delta_\theta}}{\rho^3}\Gamma^0\Gamma^2=0 \\
\begin{split}\dot{\Gamma^1} + F (\Gamma^0)^2 - 2\varepsilon a r \sin \theta \frac{\sqrt{\Delta_\theta}}{\rho^3}\Gamma^0\Gamma^3 - a^2\sin\theta\cos\theta\frac{\sqrt{\Delta_\theta}}{\rho^3}\Gamma^2\Gamma^1 \hspace{2in}\\- \varepsilon r \frac{\sqrt{\varepsilon \Delta_r}}{\rho^3}(\Gamma^2)^2 - \varepsilon r \frac{\sqrt{\varepsilon \Delta_r}}{\rho^3}(\Gamma^3)^2  = 0 \end{split} \\\begin{split}
\dot{\Gamma^2} - \varepsilon a^2 \sin\theta\cos\theta \frac{\sqrt{\Delta_\theta}}{\rho^3} (\Gamma^0)^2 - 2\varepsilon a \cos \theta \frac{\sqrt{\varepsilon \Delta_r}}{\rho^3} \Gamma^0\Gamma^3 + \varepsilon a^2 \sin\theta \cos \theta \frac{\sqrt{\Delta_\theta}}{\rho^3} (\Gamma^1)^2 \hspace{.7in} \\+ r \frac{\sqrt{\varepsilon \Delta_r}}{\rho^3} \Gamma^1 \Gamma^2 - \left(\textrm{cotan}\theta(r^2+a^2)\frac{\sqrt{\Delta_\theta}}{\rho^3} + \frac{G}{\rho} \right) (\Gamma^3)^2 = 0 \end{split} \\ & \dot{\Gamma^3} + r \frac{\sqrt{\varepsilon\Delta_r}}{\rho^3}\Gamma^1\Gamma^3 + \left(\textrm{cotan}\theta(r^2+a^2)\frac{\sqrt{\Delta_\theta}}{\rho^3} + \frac{G}{\rho} \right)\Gamma^2\Gamma^3=0\end{align*}

\section{Resultant}

\label{app_resultant}

Let $k$ be a field, and $k[X]$ denote the ring of polynomials with coefficients in $k$. If $n\in \mathbb{N}^*$, $k_n[X]$ will denote the subspace of $k[X]$ of polynomials with degree at most $n$.

Let $P,Q\in k[X]$, $n=\deg P$, $m=\deg Q$. We suppose $n>0$ and $m>0$ so that neither $P$ nor $Q$ is zero. Consider the equation: \begin{equation} UP+ VQ=0 \label{eq_gauss}\end{equation}  where $U$ et $V$ are two elements of $k[X]$.

\eqref{eq_gauss} is clearly equivalent to $UP=-VQ$. Let $D$ denote the \textrm{pgcd} of $P$ and $Q$ then $P=DP'$ and $Q=DQ'$ where $\textrm{pgcd}(P',Q')=1$. 

With these notations \eqref{eq_gauss} is equivalent to $UP'=-VQ'$, but, as $\textrm{pgcd}(P',Q')=1$ and $k[X]$ is principal, then this implies that $P'$ divides $V$. There is therefore a polynomial $C\in k[X]$ such that $V=P'C$, and so $U=-Q'C$. The set of solutions to \eqref{eq_gauss} is hence:

$$ \left\{ \left(-\frac{Q}{D}C,\frac{P}{D}C\right), C\in k[X] \right\} $$

From this, we deduce that there is a solution $(U,V) \in k_{m-1}[X] \times k_{n-1}[X]$ if and only if $\textrm{pgcd}(P,Q)\neq 1$. We can also express this in another way. Define a linear map $\phi_{P,Q}$ by:

\begin{equation} \phi_{P,Q}: \begin{array}{ccc} k_{m-1}[X]\times k_{n-1}[X]& \longrightarrow& k_{n+m-1}[X] \\ (U,V) &\longmapsto & UP+VQ\end{array} \end{equation}

According to the preceding discussion we see that, $\phi_{P,Q}$ is injective if and only if $\textrm{pgcd}(P,Q)=1$

The transpose of the matrix of $\phi_{P,Q}$ expressed in the bases \[ \left( (X^{m-1},0), \dots , (1,0),(0,X^{n-1}),\dots,(0,1) \right)\] \[(X^{m+n-1}, X^{m+n-2},\dots, X, 1)\] of $k_m[X]\times k_n[X]$ and $k_{m+n-1}[X]$ respectively is called Sylvester's matrix $S(P,Q)$ and its determinant, denoted by $R(P,Q)$, (and thus the determinant of the endomorphism $\phi_{P,Q}$) is called the resultant of $P$ and $Q$.
\begin{prop}
Let $\displaystyle P=\sum_{i=0}^{n} a_i X^i, Q= \sum_{j=0}^{m} b_j X^j $ be two polynomials with coefficients in $k$ then the Sylvester matrix $S(P,Q)$ is given by:

\begin{equation} S(P,Q)=\left( \begin{array}{llllllllllll} a_n & \dots &\dots & \dots& a_0 & 0 & \dots & \dots &0 
\\ 0 & a_n & \dots & \dots & \dots & a_0 & 0 & \dots & 0  
\\ \vdots & \ddots & \ddots & \dots & \dots & \dots&\ddots & \ddots & \vdots 
\\ 0 & \dots & \dots & 0 & a_n & \dots & \dots &\dots  & a_0 
\\ b_m &\dots& \dots &   b_0 & 0  &\dots & \dots & \dots & 0 
\\ 0 & b_m & \dots& \dots & b_0 & 0 & \dots &0 & \vdots
\\ \vdots & \ddots & \ddots  & \dots & \dots & \ddots & \ddots  & \dots & \vdots
\\ 0 & \dots & \dots & 0 & b_m & \dots & \dots & b_0 & 0 
 \\ 0 & \dots & \dots & \dots & 0 & b_m & \dots& \dots & b_0\end{array} \right)  \end{equation}
\end{prop}

From our previous discussion we have:

$$ R(P,Q) = 0 \Leftrightarrow \textrm{pgcd}(P,Q) \neq 1$$

If we move instead to an extension $L$ of $K$ containing all the roots of $P$ and $Q$, then this condition is equivalent to the fact that $P$ and $Q$ have a common root in $L$.

We recall the following result regarding the resultant:

\begin{prop}
\label{res_expression}
Let $P,Q \in k[X]$, $\deg P = n $, $\deg Q = m$.
Let $L$ be a splitting field of $P$ and $\alpha_1,\dots \alpha_n$ be the (not necessarily distinct) roots of $P$, then:

$$R(P,Q)=a_n^m \prod_iQ(\alpha_i)$$

In this formula, $a_n$ is the coefficient of $X^n$ in $P$.
\end{prop}

\begin{definition}
When $\deg P'=n-1$ (which is always the case when the characteristic of $k$ is $0$), the discriminant of $P$ is defined by: $$\Delta(P)= \frac{(-1)^{\frac{n(n-1)}{2}}}{a_n}R(P,P')$$
\end{definition}

From Proposition~\ref{res_expression} we deduce:

\begin{prop}
Let $P\in k[X]$ and suppose that $P'$ is of degree $n-1$ then, in a splitting field of $P$:
\label{discr_expression}

$$\Delta(P)=a_n^{2n-1}\prod_{i<k}(\alpha_i - \alpha_k)^2$$

Where $\alpha_1,\dots,\alpha_n$ are the (not necessarily distinct) roots of $P$.
\end{prop}

\section{Diverse useful formulae in Boyer-Lindquist like coordinates}
\label{app:diverse}
\begin{lemme}
$$g_{\phi\phi}g_{tt}-g_{\phi t}^2 = - \frac{\Delta_r\Delta_\theta \sin^2\theta}{\Xi^4}$$
\end{lemme}

\begin{lemme}
\label{lemme:ginverse}
$$(g^{ij})=\left(\begin{array}{cccc}-\frac{g_{\phi \phi} \Xi^4}{\sin^2\theta \Delta_\theta \Delta_r} & 0 & 0 & \frac{\Xi^4 g_{\phi t}}{\sin^2\theta \Delta_r \Delta_\theta} \\ 0 & \frac{1}{g_{rr}} & 0 & 0 \\ 0 & 0 & \frac{1}{g_{\theta\theta}} & 0 \\ \frac{\Xi^4 g_{\phi t}}{\sin^2 \theta \Delta_r \Delta_\theta }& 0 & 0 & -\frac{g_{tt} \Xi^4}{\sin^2\theta \Delta_r\Delta_\theta} \end{array} \right)$$
\end{lemme}
\begin{lemme}
\label{lemme:grad_t}
The metric-dual of $\dd t$ is given by:
\begin{align*} \nabla t&=\frac{\Xi^4}{\sin^2\theta \Delta_\theta \Delta_r}(-g_{\phi\phi}\partial_t + g_{\phi t} \partial_\phi) \\ 
\end{align*}
\end{lemme}
\begin{lemme}
In Boyer-Lindquist-like coordinates one can write:
\begin{align*}
g_{tt} &= \frac{1}{\Xi^2} \left( -1 +\frac{2Mr}{\rho^2} + l^2(r^2+a^2\sin^2\theta)\right) \\
g_{\phi t} &= -\frac{a\sin^2\theta}{\Xi^2}\left( l^2(r^2+a^2) + \frac{2Mr}{\rho^2}\right) \end{align*}
\end{lemme}
\section{ Gluing topological spaces}
\label{app:gluing}
Let $X$ and $Y$ be two topological spaces, $U$ and $V$ be open subsets of $X$ and $Y$ respectively and $\phi$ be a homeomorphism of $U$ onto $V$. We outline here the construction of a new topological space containing both $X$ and $Y$ and where $U$ and $V$ have been identified. In a sense, we will have glued $X$ to $Y$ along $U$ and $V$.  Let $X\coprod Y$ denote their coproduct (or disjoint union) and $i : X \longrightarrow X\coprod Y, j: Y: \longrightarrow X\coprod Y$ the canonical injections.
 Define an equivalence relation on $X\coprod Y$ by:

\begin{equation} \label{eq:equiv_rel} p \sim q \Leftrightarrow \left( [p=q] \text{ or } [p = i(x), q=j(\phi(x)), x \in U] \text{ or } [q=i(x), p=j(\phi(x)),x \in U] \right) \end{equation}

Denote by $X\coprod_\phi Y$ the quotient space of $X\coprod Y$ by this equivalence relation and $\pi: X\coprod Y \longrightarrow X\coprod_\phi Y$ the canonical projection.
We quote without proof the following results:
\begin{prop}
\label{prop:gluing1}
\begin{enumerate}
\item $\bar{j}=\pi\circ j, \bar{i}=\pi\circ i$ are continuous injective and open maps. $X$ and $Y$ can then be identified with the open subsets $\bar{i}(X)$ and $\bar{j}(Y)$ of $X\coprod_\phi Y$.

\item $\bar{i}(X)\cap \bar{j}(Y)= \bar{i}(U)=\bar{j}(V)$

\item If $F$ is an arbitrary topological space, $f : X\coprod_\phi Y \rightarrow F$ is continuous if and only if the maps $f\circ \bar{i}$ et $f\circ \bar{j}$ are.
\item $\pi$ is an open map
\end{enumerate}
\end{prop}

Points 2 and 3 can be useful for constructing maps on $ X\coprod_\phi Y$ from maps $f,g$ defined on $X$ and $Y$ separately. In fact, it suffices that they satisfy $f(x)=g(\phi(x))$ for every $x\in U$ for them to piece together to form a well-defined continuous map on $X\coprod_\phi Y$. This is sometimes called the mapping lemma; it has natural generalisations to maps and manifolds with more regularity. 
The above proposition also serves to prove the following results:

\begin{prop}
\begin{enumerate}
\item If $X$ and $Y$ are both locally Euclidean, then $X\coprod_\phi Y$ is too.
\item If $X$ and $Y$ are both second-countable, $X\coprod_\phi Y$ is too.
\end{enumerate}
\end{prop}

It is well known that separation properties of a quotient are relatively independent of the separation properties of the original space, however since the canonical projection map is open one has the following result:
\begin{lemme} 
\label{lemme:separation1}
$X\coprod_\phi Y$ is Hausdorff if and only if $R=\{(p,q) \in (X\coprod Y)^2, p\sim q\}$ is closed in $(X\coprod Y)^2$
\end{lemme}

With this result we can prove a technical criterion that will guarantee separation in all cases of interest in the text:

\begin{lemme} 
Suppose that $X$ and $Y$ are Hausdorff and first countable then if there is no sequence $(x_n)_{n\in\mathbb{N}}$ of points in $U$ converging to a point in $\bar{U}\setminus U$ and such that $\phi(x_n)_{n\in \mathbb{N}}$ converges to a point in $\bar{V}\setminus V$, $X\coprod_\phi Y$ is Hausdorff.
\end{lemme}

\begin{proof}
By Lemma~\ref{lemme:separation1} it suffices to show that $R=\{(p,q) \in (X\coprod Y)^2, p\sim q\}$ is closed in $(X\coprod Y)^2$. Furthermore, as $X$ and $Y$ are first countable, it suffices to show that if two sequences $(p_n)_{n\in\mathbb{N}}$ and $(q_n)_{n\in \mathbb{N}}$ of points in $X\coprod Y$ are such that $\forall n\in\mathbb{N}, p_n \sim q_n$ and $p_n \underset{n \to \infty}{\longrightarrow} p, q_n \underset{n \to \infty}{\longrightarrow} q$ then $p\sim q$.

Let $(p_n)_{n\in \mathbb{N}}$ and $(q_n)_{n\in \mathbb{N}}$ be two such sequences. We can restrict ourselves to the case where $p \in i(X)$ and $q\in j(Y)$ as $p$ and $q$ play symmetric roles and if $p \in i(X)$ (resp. $j(Y)$) then $q\in i(X)$ (resp. $j(Y)$) then for all large enough $n$, $p_n \in i(X)$ and $q_n \in i(X)$, as $i(X)$ is open in $X\coprod Y$, hence:
$$ \exists N \in \mathbb{N}, \forall n \geq N, p_n=q_n \Rightarrow p=q$$

Assume now that $p\in i(X)$ and $ q \in j(Y)$, we distinguish 3 cases:

\begin{itemize}
\item[Case 1:] $p \in i(X)\setminus \overline{i(U)}$, then there is $N\in\mathbb{N}$ such that $\forall n \geq N, p_n \in i(X) \setminus \overline{i(U)}$, but as $q_n\sim p_n$ for every $n\in \mathbb{N}$ it follows that for all $n\geq N,p_n=q_n$ so $p=q$. Which is excluded as $i(X) \cap j(Y) = \emptyset$
\item[Case 2:] $p \in i(U)$, then again, there is $N \in \mathbb{N}$ such that $\forall n \geq N, p_n \in i(U)$. Since $q \in j(Y)$ there is also $N' \in \mathbb{N}$ such that $\forall n \geq N', q_n \in j(Y)$. Moreover,  as for every $n\in \mathbb{N}, p_n \sim q_n$ it follows from~\eqref{eq:equiv_rel} that:

$$\forall n \geq \max(N,N'), \left\{\begin{array}{c} q_n = j(y_n), y_n \in V\\  p_n=i(x_n), x_n \in U\\ y_n=\phi(x_n) \end{array}\right.$$

As $i$ and $j$ are homeomorphisms onto their ranges, the sequences $(x_n)$ and $(y_n)$ converge to points $x\in X$ and $y\in Y$ respectively. Furthermore, $\phi$ being continuous, one must have $y=\phi(x)$ so: $p\sim q$

\item[Case 3:] $p\in \overline{i(U)}\setminus i(U)$, if only a finite number of points of the sequence lie in $i(U)$ then there is a rank $N$ above which $q_n=p_n$ so $q=p$ which is excluded as $q\in j(Y)$. Thus, we can assume that one can extract a subsequence $(p_{\varphi(n)})_{n\in \mathbb{N}}$ of $(p_n)_{n\in\mathbb{N}}$ such that for all $n\in \mathbb{N}$, $p_{\varphi(n)}\in i(U)$. Necessarily, $q\in \overline{j(V)}$, but $q\not\in j(V)$ as this would imply $p \in i(U)$, so $q\in \overline{j(V)}\setminus {j(V)}$. However, as $\forall n \in \mathbb{N}, q_n \sim p_n$ there must exist sequences $(x_n)$ and $(y_n)$ of points of $X$ and $Y$ respectively such that $(x_n)$ converges to a point in $\bar{U}\setminus U$, $(y_n)$ to a point in $\bar{V} \setminus V$ and $y_n=\phi(x_n)$ for sufficiently large $n$, but this contradicts our hypothesis.
Hence $p \sim q$ and $R$ is closed.
\end{itemize}
\end{proof}

\label{app:kruskal_metric}
\end{document}